\documentclass[runningheads,envcountsame]{fundam}

\publyear{22}
\papernumber{2102}
\volume{185}
\issue{1}

\usepackage[T1]{fontenc}
\usepackage{amsmath,amssymb, mathtools}
\usepackage{array}
\usepackage{color}
\usepackage[shortlabels]{enumitem}
\usepackage[colorlinks=true,allcolors=blue]{hyperref}
\usepackage{leftindex}
\usepackage[kerning]{microtype}
\usepackage{tikz}

\makeatletter
\newcounter{dummy}
\newcommand\myitem[1][]{\item[#1]\refstepcounter{dummy}\def\@currentlabel{#1}}
\makeatother

\tikzstyle{every picture} = [scale=.45]
\tikzstyle{every node} = [draw, fill=white, circle, inner sep=0pt, minimum size=4pt]
\tikzstyle{d} = [very thick]
\tikzstyle{i} = [draw, fill=black, circle, inner sep=0pt, minimum size=4pt]
\tikzstyle{n} = [draw=none, rectangle, inner sep=2pt]

\let\:\colon

\let\ld\backslash
\let\le\leqslant
\let\meta\looparrowright

\let\rd/
\let\setminus\smallsetminus
\let\subset\subseteq
\let\sle\sqsubseteq
\let\jle\preccurlyeq
\let\jl\prec

\newcommand\alg{\mathrm{alg}}
\newcommand\po{\mathrm{po}}
\newcommand\lscr[2]{\leftindex_{#1}{#2}}
\newcommand{\m}{\mathbf}

\newcommand\nein{\mathord{\sim}}
\newcommand\no{\mathord{-}}

\newcommand\pair[1]{\langle#1\rangle}

\DeclareMathOperator\Idp{E^+\!}
\DeclareMathOperator\IDP{\mathbf{E}^+\!}
\DeclareMathOperator\ZIdp{ZE^+\!}
\DeclareMathOperator\ZIDP{\mathbf{ZE}^+\!}

\newcommand\category[1]{\expandafter\newcommand\csname#1\endcsname{\textsf{#1}}}
\category{Meta}
\category{Alg}
\category{Pos}
\category{Set}
\category{PBrA}

\begin{document}

\title{Balanced residuated partially ordered semigroups}

\address{1 University Drive, Chapman University, Orange CA 92866, USA}

\author{Stefano Bonzio\\
University of Cagliari, Italy\\
stefano.bonzio{@}unica.it
\and Jos\'e Gil-F\'erez\\
Chapman University, Orange CA 92866, USA\\
gilferez{@}chapman.edu
\and Peter Jipsen\\
Chapman University, Orange CA 92866, USA\\
jipsen{@}chapman.edu
\and Adam P\v{r}enosil\\
University of Barcelona, Spain\\
adam.prenosil{@}gmail.com
\and Melissa Sugimoto\\
Chapman University, Orange CA 92866, USA\\
msugimoto{@}chapman.edu}

\maketitle

\runninghead{S. Bonzio, J. Gil-F\'erez, P. Jipsen, A. P\v{r}enosil, M. Sugimoto}{Balanced residuated semigroups}

\begin{abstract}
A \emph{residuated semigroup} is a structure $\pair{A,\le,\cdot,\ld,\rd}$ where $\pair{A,\le}$ is a poset and $\pair{A,\cdot}$ is a semigroup such that the residuation law $x\cdot y\le z\iff x\le z/y\iff y\le x\ld z$ holds. An element $p$ is \emph{positive} if $a\le pa$ and $a \le ap$ for all~$a$. A residuated semigroup is called \emph{balanced} if it satisfies the equation $x \ld x \approx x \rd x$ and moreover each element of the form $a \ld a = a \rd a$ is positive, and it is called \emph{integrally closed} if it satisfies the same equation and moreover each element of this form is a global identity. We show how a wide class of balanced residuated semigroups (so-called \emph{steady} residuated semigroups) can be decomposed into integrally closed pieces, using a generalization of the classical P\l{}onka sum construction. This generalization involves gluing a disjoint family of ordered algebras together using multiple families of maps, rather than a single family as in ordinary P\l{}onka sums.
\end{abstract}


\section{Introduction}

A \emph{residuated partially ordered semigroup}, or \emph{residuated semigroup} for short, is a structure of the form $\m A = \pair{A,\le,\cdot,\ld,\rd}$ such that $\pair{A,\le}$ is a poset, $\pair{A,\cdot}$ is a semigroup, and the operations $\ld$ and $\rd$ are the \emph{residuals} of the \emph{multiplication} $\cdot$, that is
\[
x\cdot y \le z \quad\iff\quad x \le z\rd y \quad\iff\quad y \le x\ld z.
\]
As usual, we denote the product $x\cdot y$ simply by $xy$. The residuation equivalences imply that $\cdot$ is order-preserving in both arguments and that $\ld$ and $\rd$ are order-preserving in the numerator and order-reversing in the denominator. If $\pair{A,\le}$ is a join-semilattice or a meet-semilattice, then $\le$ is term-definable, hence residuated join-semilattices (or residuated meet-semilattices) form a variety. However, residuated semigroups are only a po-variety, i.e., they are defined by inequations and all fundamental operations are either order-preserving or order-reversing in each argument~\cite{Pi04}.

An element $u$ of a residuated semigroup $\m A$ is called a \emph{global identity} if $u\cdot a = a = a\cdot u$ for all $a\in A$. If $\m A$ contains a global identity, then it is unique and we usually represent it by $1$. We call the structure $\widehat{\m A} = \pair{A,\le,\cdot,1,\ld,\rd}$ a \emph{residuated partially ordered monoid}, or \emph{residuated monoid} for short.\footnote{In~\cite{BGJPS2024} the terminology ``residuated poset'' was used instead of residuated monoid, to emphasize the similarity with residuated lattices (where the partial order is in fact a lattice).}

Examples of residuated semigroups and residuated monoids include all groups (ordered by the antichain order), partially ordered groups, hoops, Brouwerian semilattices and generalized Boolean algebras, as well as all subreducts of residuated lattices. In the case of groups with the antichain order, the residuals are $x\ld y=x^{-1}y$ and $x\rd y=xy^{-1}$. Under this interpretation of the residuals, groups satisfy the identity $x\ld x\approx x\rd x$, which can fail in relation algebras and complex algebras of groups.
Residuated monoids also satisfy $y\le (x\rd x)y$ and $y\le y(x\rd x)$. Residuated semigroups in which these three (in)equations hold are called \emph{balanced}. They are the main focus of our investigations.

The representation theory presented in this paper rests on partitioning a balanced residuated semigroup $\m A$ into a family of sets indexed by the positive ($a \le pa$ and $a \le ap$ for all $a \in A$) idempotents $p$ of $\m A$, namely $A_p := \{ a \in A : a \ld a = p \} = \{ a \in A : a \rd a = p \}$. We focus on the case where these partition classes are universes of subalgebras of $\m A$, i.e., closed under multiplication and residuation. In that case, each $\m A_p$ has $p$ as its global identity, which is the only positive idempotent of~$\m A_p$. Moreover, the positive idempotents form a join semilattice $\IDP \m A := \langle \Idp \m A, \cdot \rangle$, since each positive idempotent $p$ of a balanced residuated semigroup is central (that is, $p a = a p$ for all $a \in A$). In other words, we obtain a family of (reducts of) residuated monoids $\m A_p$ indexed by the join semilattice $\IDP \m A$. This decomposition is the topic of Section~\ref{sec:decompositions}.

The question then arises of how to reconstruct the residuated semigroup $\m A$ from the ordered subalgebras $\m A_p$ indexed by the join semilattice $\IDP \m A$. The appropriate tool for this purpose are so-called P\l{}onka sums of semilattice directed systems of homomorphisms, which are systems consisting of a family of algebras $\m A_p$ indexed by a join semilattice $\m I$ together with a suitable family of homomorphisms $\varphi_{ij}\colon \m A_i \to \m A_j$. The main use of this classical algebraic construction~\cite{Plonka67,Plonka68a,Plonka68b} is to provide a structural description of so-called regular varieties (see e.g.~\cite{PR92,BPP22}). However, it has also been successfully applied to various classes of residuated structures, including:
\begin{itemize}
\item even and odd involutive commutative residuated chains~\cite{Je22},
\item commutative idempotent involutive residuated lattices~\cite{JiTuVa21},
\item locally integral involutive po-monoids in~\cite{GFJiLo23}, and
\item locally integral involutive po-semigroups in~\cite{GFJiSu24}.
\end{itemize}
In~\cite{GFJiLo23,GFJiSu24} involutivity means that the residuals are definable from two unary negation operators $\nein,\no$ as $x\ld y:=\nein(\no y\cdot x)$ and $x\rd y:=\no(y\cdot \nein x)$. This property is then inherited by the components $\m A_p$. Local integrality further ensures that the components $\m A_p$ are \emph{integral}, i.e., have the global identity as their top element. In the current paper we obtain similar structural results for a much larger class of residuated semigroups, where the components need not be involutive or integral. Instead, the components are merely assumed to be \emph{integrally closed}, i.e.\ to satisfy $x \ld x \approx 1$, or equivalently $x \rd x \approx 1$.

The present context in fact calls for a generalization of the P\l{o}nka sum construction, which applies to what we call \emph{semilattice directed systems of metamorphisms}, where the family of homomorphisms $\varphi_{ij}$ is replaced by multiple families of maps (which are not necessarily homomorphisms). This construction was already introduced in the conference paper~\cite{BGJPS2024}, which the present paper expands on. The construction of the P\l{}onka sum of metamorphisms is described in detail in Sections~\ref{plonka} and~\ref{sec:sums:of:posets}, which deal respectively with the algebraic structure and the order structure of the P\l{}onka sum.

Let us now sketch how this construction applies to the family $\m A_p$ of residuated semigroups obtained by decomposing a sufficiently well-behaved balanced residuated semigroup $\m A$. This is explained in more detail and greater generality in Section~\ref{sec:fibrant}. Given $p \le q$ in $\IDP \m A$, we define the maps $\varphi_{pq}, \psi_{pq}\colon A_p \to A_q$ as $\varphi_{pq}(a) := a q$ and $\psi_{pq}(a) := a / q$. These maps define a semilattice directed system of metamorphisms over the semilattice $\IDP \m A$. The question is now, given $a \in A_p$ and $b \in A_q$, how does the P\l{}onka sum construction use this data to compute the elements $a \cdot b, a \ld b, a \rd b$ in~$\m A$ and to determine whether $a \le b$ in $\m A$? We do so in the following way: for $r := p \cdot q$
\[
  a \cdot b = \varphi_{pr}(a) \cdot_{r} \varphi_{qr}(b), \qquad
  a \ld b = \varphi_{pr}(a) \ld_r \psi_{qr}(b), \qquad
  a \rd b = \psi_{pr}(a) \rd_{\!r} \varphi_{qr}(b),
\]
  and
\[
  a \le b \iff \varphi_{pr}(a) \le_{r} \psi_{qr}(b),
\]
  where the subscript $r$ indicates that the operations and the order are computed in $\m A_r$. The main results of this paper, namely Theorems~\ref{thm:steady:over:I:partition:meta} and~\ref{thm:construction:meta}, show that each sufficiently well-behaved (``steady'') residuated semigroup can be decomposed in this way into a family of integrally closed residuated monoids, and conversely that such P\l{}onka sums of integrally closed residuated monoids yield steady residuated semigroups.

The case of idempotent balanced residuated semigroups is discussed in more detail in Section~\ref{sec:idempotent case}. The paper concludes with some further examples of the above constructions in Section~\ref{sec:examples}. We summarize the main classes of ordered algebras considered in this paper in Table~\ref{tab:residuated_structures_summary} in Section~\ref{sec:conclusion}.

\section{Positive Idempotents in Residuated Semigroups}

An element $p$ of a residuated semigroup $\m A$ is \emph{positive} if for every $a\in A$, we have both $a\le pa$ and $a\le ap$, or equivalently, if for every $a \in A$, we have both $p \ld a \le a$ and $a \rd p \le a$. To see that these two conditions are equivalent, 
notice that 
$p\ld a \le p(p\ld a) \le a$ and $a\rd p \le (a \rd p) p \le a$ if $b \le p b$ and $b \le b p$ for every $b \in A$, 
while $a \le p \ld (pa) \le pa$ and $a \le (ap) \rd p \le a p$ if $p \ld b \le b$ and $b \rd p \le b$ for every $b \in A$. Moreover, the set $A^+$ of positive elements of $\m A$ is an upset (upward-closed subset) of $\pair{A,\le}$, because if $p\in A^+$ and $p\le q$, then for every $a\in A$, $a \le pa \le qa$ by the monotonicity of the multiplication and $a\le aq$ analogously. If $\m A$ has a global identity $1$, then $A^+$ is the principal upset generated by $1$, that is, $A^+ = \{a\in A : 1 \le a\}$.

An element $p\in A$ is \emph{idempotent} if $p^2 = p\cdot p = p$. We denote by $\Idp\m A$ the set of all positive idempotents of $\m A$. If $\m A$ has a global identity $1$, then $1$ is the smallest element of $\Idp\m A$. For every element $p\in\Idp\m A$, we define the following sets:
\begin{align*}
A\rd p &:= \{a\rd p : a\in A\}, & Ap &:= \{ap : a\in A\}, & A_p &:= \{a\in A : a\ld a = p\},\\
p\ld A &:= \{p\ld a : a\in A\}, & pA &:= \{pa : a\in A\}, & \lscr pA &:= \{a\in A : a\rd a = p\}.
\end{align*}

\begin{lemma}\label{lem:char:A/p=Ap}
The following equalities hold for every $p\in\Idp\m A$.
\begin{align*}
(1)\quad    A\rd p &= \{a\in A : a = a\rd p\} = \{a\in A : a\le a\rd p\} = \{a\in A : ap\le a\}\\
 &= \{a\in A : ap = a\} = Ap.\\
(2)\quad    p\ld A &= \{a\in A : a = p\ld a\} = \{a\in A : a\le p\ld a\} = \{a\in A : pa\le a\}\\
 &= \{a\in A : pa = a\} = pA.
\end{align*}
\end{lemma}

\begin{proof}
We will only prove (1), since the proof of (2) is completely analogous. First of all, notice that the three middle equalities are consequences of residuation and the fact that $p$ is positive and therefore $a\rd p \le a$ and $a \le ap$. Now, for every $b\in A$, we have that $b\rd p = b\rd p^2 = (b\rd p)\rd p$, which shows that $A\rd p \subseteq \{a\in A : a = a\rd p\}$. The reverse inclusion is trivial. We also have that for every $b\in A$, the equality $bpp = bp$ implies that $Ap\subset \{a\in A : ap =a\}$. The reverse inclusion is also trivial.
\end{proof}

\begin{lemma}\label{lem:Ap:closure:meets:joins}
Let $\m A$ be a residuated semigroup and $p\in\Idp\m A$.
\begin{enumerate}[(1)]
\item The maps $a \mapsto pa$ and $a \mapsto ap$ are join-preserving closure operators whose images are $pA$ and $A p$, respectively.
\item The maps $a \mapsto p\ld a$ and $a \mapsto a\rd p$ are meet-preserving interior operators whose images are $p\ld A$ and $A\rd p$, respectively.
\item The sets $pA = p\ld A$ and $Ap = A\rd p$ are closed under existing meets and joins.
\end{enumerate}
\end{lemma}

\begin{proof}
(1) Consider the map $\gamma\: a\mapsto pa$. The positivity and idempotence of $p$ imply that $a \le \gamma(a)$ and $\gamma(\gamma(a)) = \gamma(a)$ for all $a\in A$, respectively. And if $a\le b$, then $\gamma(a) \le \gamma(b)$, by the monotonicity of the product. Hence, $\gamma$ is a closure operator. Its image is $\{\gamma(a) : a\in A\} = pA$. Since multiplication preserves all existing joins \cite[Thm 3.10]{GJKO2007} we have $\gamma(\bigvee X) = \bigvee\gamma[X]$ whenever $\bigvee X$ exists. The proof of the other claim is analogous.

(2) Consider now the map $\delta\: a\mapsto p\ld a$. The positivity and idempotence of $p$ now imply that $\delta(a)\le a$ and $\delta(\delta(a)) = \delta(a)$, and the monotonicity in the numerator of the left residual implies the monotonicity of $\delta$. Hence, $\delta$ is an interior operator. Since residuals preserve all existing meets in the numerator \cite[Thm 3.10]{GJKO2007} we have $\delta(\bigwedge X) = \bigwedge\delta[X]$ whenever $\bigwedge X$ exists. The proof of the other claim is analogous.

(3) Recall that the images of closure and interior operators are closed under existing meets and existing joins, respectively. Hence, since $pA = p\ld A$, by Lemma~\ref{lem:char:A/p=Ap}, this set is closed under existing meets and joins, by the previous parts.
\end{proof}

\begin{lemma}\label{lem:IdpA:identities}
Let $\m A$ be a residuated semigroup and $a\in A$.
\begin{enumerate}[(1)]
    \item If $a\rd a$ is positive, then it is the largest $p\in\Idp\m A$ such that $a\in p\ld A$.
    \item If $a\ld a$ is positive, then it is the largest $p\in\Idp\m A$ such that $a\in A\rd p$.
\end{enumerate}
\end{lemma}

\begin{proof}
Notice that $(a\rd a)a \le a$ holds by residuation, and since $(a\rd a)(a\rd a)a \le (a\rd a)a \le a$, we also have that $(a\rd a)(a\rd a) \le a\rd a$. If moreover $a\rd a$ is positive, then we also have the inequality $a\rd a \le (a\rd a)(a\rd a)$. Hence, $p = a\rd a \in\Idp\m A$ and $pa \le a$, that is, $a\in p\ld A$. Consider $q\in \Idp\m A$ such that $a\in q\ld A$. Then $qa\le a$ and, by residuation, $q\le a\rd a$. The proof of the second claim is analogous.
\end{proof}

The following lemma is an immediate consequence of this.

\begin{lemma}\label{lem:char:interdefinability}
For every $p\in\Idp\m A$, we have the following equalities.
\[
A\rd p = \smashoperator{\bigcup\limits_{p\le q\in\Idp\m A}} A_q,
\quad A_p = (A\rd p)\setminus \smashoperator{\bigcup\limits_{p < q\in\Idp\m A}} A_q,
\quad p\ld A = \smashoperator{\bigcup\limits_{p\le q\in\Idp\m A}} \lscr qA,
\quad \lscr pA = (p\ld A)\setminus \smashoperator{\bigcup\limits_{p < q\in\Idp\m A}} \lscr qA.
\]
\end{lemma}

\begin{proof}
If $a\in A\rd p$, then $ap\le a$, and therefore $p\le a\ld a$, whence we deduce that $q = a\ld a$ is positive. By Lemma~\ref{lem:IdpA:identities}, $q\in\Idp\m A$, and $a\in A_q$ trivially. In order to prove the reverse inclusion, suppose that $p\le q\in\Idp\m A$ and $a\in A_q$. Then, $p \le q = a\ld a$ and by residuation $ap\le a$, that is, $a\in A\rd p$.

The second equality follows from the first one and the fact that the family $\{A_q : q\in\Idp\m A\}$ consists of disjoint sets. The other two equalities are proven analogously.
\end{proof}

An element $p$ of a residuated semigroup $\m A$ is \emph{central} if $pa = ap$ for every $a\in A$. A residuated semigroup is \emph{commutative} if all its elements are central, that is, if it satisfies the equation $xy\approx yx$. The following proposition characterizes the centrality of the positive idempotents in a residuated semigroup.

\begin{proposition}\label{prop:equivalent:char:idempositive:central}
The conditions below are equivalent for every residuated semigroup~$\m A$ and every $p\in\Idp\m A$.
\begin{enumerate}[(1)]
    \item For every $a\in A$, $pa = ap$, that is, $p$ is central.
    \item For every $a\in A$, $pa = a$ if and only if $ap = a$.
    \item For every $a\in A$, $p\ld a = a$ if and only if $a\rd p = a$.
    \item For every $a\in A$, $p\ld a = a\rd p$.
    \item $p\ld A =  A\rd p$.
\end{enumerate}
\end{proposition}

\begin{proof}
\begin{itemize}[topsep=0pt, align=left, left=0pt .. \parindent]
\item[(1) $\Leftrightarrow$ (2)] The fact that~(1) implies~(2) is trivial. In order to prove the reverse implication, notice that $p(pa) = pa$ and therefore, by~(2) we also have that $(pa)p = pa$. Analogously, we can also show that $p(ap) = ap$. Hence, $pa = pap = ap$.

\item[(2) $\Leftrightarrow$ (3) $\Leftrightarrow$ (5)] This follows from Lemma~\ref{lem:char:A/p=Ap}.

\item[(3) $\Leftrightarrow$ (4)] Notice that $p\ld (p\ld a) = (pp)\ld a = p\ld a$ and therefore we deduce by~(3) that $(p\ld a)\rd p = p\ld a$. Analogously, we can obtain that $(p\ld a)\rd p = p\ld a$. hence, $p\ld a = p\ld a\rd p = a\rd p$. The reverse implication is trivial.\QED
\end{itemize}\let\QED\relax
\end{proof}

It is easy to see that the set $\ZIdp\m A$ of central positive idempotents of $\m A$ is closed under multiplication and therefore (if non-empty) forms a join-semilattice $\ZIDP\m A = \pair{\ZIdp\m A,\cdot}$, since the restriction of the multiplication of $\m A$ to $\ZIdp\m A$ is associative, idempotent, and commutative. Moreover, notice that for every $p,q\in\ZIdp\m A$ we have that if $p\le q$, then $q\le pq \le qq = q$, and if $pq = q$, then $p\le pq = q$. Thus the order induced by the product of $\IDP\m A$ (viewed as a join operation) is a restriction of the order of~$\m A$. The following is an immediate consequence of the previous proposition and lemma.

\begin{corollary}\label{cor:e=ze}
The following are also equivalent for every residuated semigroup $\m A$.
\begin{enumerate}[(1)]
    \item $\Idp\m A = \ZIdp\m A$.
    \item $\lscr pA = A_p$, for all $p\in\Idp\m A$.
\end{enumerate}
\end{corollary}

\begin{example}
The product of positive idempotents of a residuated semigroup is not commutative in general. Moreover, for an arbitrary residuated semigroup $\m A$, the set $\Idp\m A$ doesn't need to be closed under multiplication. The residuated semigroup $\m A$ of Figure~\ref{fig:example:Idp:not:commutative} with three positive idempotents $\Idp\m A = \{p,q,r\}$ is an example of the failure of these properties. The element $a$ is positive, but not idempotent. Indeed, $aa = pqpq = prq = rq = r$.
\end{example}

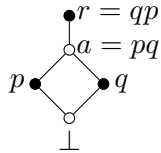
\begin{figure}[ht]\centering
\begin{tikzpicture}[baseline=0pt]
\node(4) at (0,3)[i, label=right:{$r = qp$}]{};
\node(3) at (0,2)[label=right:{$a = pq$}]{} edge (4);
\node(2) at (1,1)[i, label=right:$q$]{} edge (3);
\node(1) at (-1,1)[i, label=left:$p$]{} edge (3);
\node(0) at (0,0)[label=below:$\bot$]{} edge (1) edge (2);
\end{tikzpicture}

\caption{The product of idempotent positives is not commutative in general.}
\label{fig:example:Idp:not:commutative}
\end{figure}

\begin{remark}\label{rem:quotient:of:positives}
It is interesting to notice that for all $p,q\in\Idp\m A$, if $p\le q$, then $q\rd p = q$. Indeed, we have that $q\rd p \le q$ because $p$ is positive, and $q\le q\rd p$ because $qp = q$. Analogously, $p\ld q = q$.
\end{remark}

\section{Decompositions of Balanced Residuated Semigroups}
\label{sec:decompositions}

As we already mentioned above, given a residuated semigroup $\m A$, the families $\{A_p : p\in\Idp\m A\}$ and $\{\lscr pA : {p\in\Idp\m A}\}$ are disjoint by definition and their members are nonempty, because $p\in A_p$ and $p\in \lscr pA$. But, in general, they are not partitions of $A$, since they need not cover all of $A$. They would be so if all the \emph{self-residuals}, i.e., elements of the form $a\ld a$ or $a\rd a$, were positive. This is not an uncommon situation: for instance, if $\m A$ contains a global identity, then all its self-residuals are positive. Reciprocally, if all self-residuals of $\m A$ are positive and $\Idp\m A$ has a smallest element, then it is a global identity of $\m A$.

\begin{lemma}\label{lem:self-residuals:positive}
The conditions below are equivalent in any residuated semigroup~$\m A$.
\begin{enumerate}[(1)]
    \item $\m A$ satisfies $y \le (x\rd x) y$ and $y \le y (x\ld x)$.
    \item All self-residuals in $\m A$ are positive.
    \item $\Idp\m A = \{a\rd a : a\in A\} = \{a\ld a : a\in A\}$.
    \item The families $\{\lscr pA : p\in\Idp\m A\}$ and $\{A_p : p\in\Idp\m A\}$ are partitions of~$A$.
\end{enumerate}
\end{lemma}

\begin{proof}
\begin{itemize}[topsep=0pt, align=left, left=0pt .. \parindent]
\item[(1) $\Leftrightarrow$ (2)] All we have to show is that the inequalities of (1) imply $y \le y(x\rd x)$ and $y \le (x\ld x)y$. Using the first inequality of (1) and the residuation properties, we have
\[
(x\rd x)\ld (x\rd x) \le (x\rd x)\big((x\rd x)\ld (x\rd x)\big) \le x\rd x \le y\ld(y(x\rd x)).
\]
By the second inequality of (1) and the residuation properties applied to the displayed inequality, we obtain $y \le y \big((x\rd x)\ld (x\rd x)\big) \le y(x\rd x)$. The proof that $y \le (x\ld x)y$ is analogous.

\item[(2) $\Rightarrow$ (3)] We have already discussed that if $p\in\Idp\m A$, then $p = p\rd p$. Additionally, if all the self-residuals are positive, then they are also idempotent, by Lemma~\ref{lem:IdpA:identities}. Thus, $\Idp\m A = \{a\rd a : a\in A\}$. The other equality is proven in the same way.

\item[(3) $\Rightarrow$ (4)] For every $a\in A$, the elements $p = a\ld a$ and $q = a\rd a$ are in $\Idp\m A$ and $a\in A_p$ and $a\in \lscr qA$, whence the result follows.

\item[(4) $\Rightarrow$ (2)] If $\{A_p : p\in\Idp\m A\}$ is a partition of $A$, then $a\in A_p$, for some $p\in\Idp\m A$, that is, $a\ld a = p \in\Idp\m A$ and in particular $a\ld a$ is positive. Analogously, if $\{\lscr pA : p\in\Idp\m A\}$ is a partition of $A$, then for every $a\in A$, $a\rd a$ is positive.\QED
\end{itemize}\let\QED\relax
\end{proof}

A residuated semigroup is \emph{balanced} if it satisfies the identity $x\ld x \approx x\rd x$ and all its self-residuals are positive. The following proposition gives a number of conditions equivalent to being balanced.

\begin{proposition}\label{prop:equivalent:char:balanced}
The following conditions are equivalent for any residuated semigroup~$\m A$.
\begin{enumerate}[(1)]
    \item $\m A$ is balanced.
    \item $\mathbf{A}\models x\ld x \approx x\rd x$ and $\mathbf{A}\models y\le (x\rd x)y$ and $\mathbf{A}\models y\le y(x\ld x)$.
    \item $\m A\models (x\ld x)y\approx y(x\ld x)$ and $\m A\models y\le y(x\ld x)$. 
    \item $\m A\models (x\rd x)y\approx y(x\rd x)$ and $\m A\models y\le (x\rd x)y$.    
\end{enumerate}
\end{proposition}

\begin{proof}
First, notice that, by virtue of Lemma~\ref{lem:self-residuals:positive}, condition~(2) is equivalent to $\m A$ being balanced. Next, we will show that condition~(3) also implies that all the self-residuals are positive. Take arbitrary elements $a,b\in A$. Using the first equation of condition~(3), we have that $(a\ld a)a = a(a\ld a) \le a$, and therefore $a\ld a \le a\rd a$. Multiplying by $b$ and applying the equation and the inequation, we obtain
\[
b \le b(a\ld a) = (a\ld a)b \le (a\rd a)b.
\]
That is, $\m A$ also satisfies $y\le (x\rd x)y$ and therefore all the self-residuals are positive by Lemma~\ref{lem:self-residuals:positive}.

In order to show that~(2) and~(3) are equivalent, take arbitrary $a,b\in A$ and let $p = a\ld a$, which is positive in both cases, and therefore $p\in\Idp\m A$. If $\m A$ satisfies $x\ld x \approx x\rd x$, we have $\lscr pA = A_p$ and, by Corollary~\ref{cor:e=ze}, $pb = bp$. Thus, $\m A$ satisfies $(x\ld x)y \approx y(x\ld x)$. And if $\m A$ satisfies $(x\ld x)y \approx y(x\ld x)$, then $pb=bp$ and, by Corollary~\ref{cor:e=ze}, we have $\lscr pA = A_p$. Hence, $\m A$ satisfies $x\ld x \approx x\rd x$. The proof that conditions~(2) and~(4) are equivalent is analogous.
\end{proof}

In what follows, we will denote the term $x\rd x$ by $1_x$. Notice that, as we argued in the proof of Lemma~\ref{lem:IdpA:identities}, in every residuated semigroup $\m A$, $p\rd p = p$ for every $p\in\Idp\m A$, that is, $1_p = p$. Hence, if all the self-residuals in a residuated semigroup are positive, then $1_{1_x}\approx 1_x$ holds, as well as $1_x\cdot x\approx x$. The following results describe many more properties of the terms of the form $1_x$ in residuated semigroups where all these terms are central and, in particular, in all balanced residuated semigroup. The term $1'_x = x\ld x$ has similar properties. 

\begin{lemma}\label{lem:inequalities:selfresiduals:central}
The following inequations hold in all residuated semigroups that satisfy $1_x\cdot y \approx y\cdot 1_x$.
\[
1_x \le 1_{xy},\quad
1_y \le 1_{xy},\quad
1_x \le 1_{x\rd y},\quad
1_y \le 1_{x\rd y},\quad
1_x \le 1_{y\ld x},\quad
1_y \le 1_{y\ld x}.
\]
As a consequence, the following inequations also hold.
\[
1_x\cdot 1_y \le 1_{xy},\quad 1_x\cdot 1_y\le 1_{x\rd y},\quad 1_x\cdot 1_y \le 1_{y\ld x}.
\]
\end{lemma}

\begin{proof}
For all $a,b\in A$, the inequality $1_a\cdot ab \le ab$ implies that $1_a\le ab\rd ab = 1_{ab}$. Analogously, $1_b\cdot ab = a\cdot 1_b\cdot b \le ab$ implies that $1_b\le ab\rd ab = 1_{ab}$. Also, $1_a(a\rd b)b\le 1_a \cdot a \le a$ and therefore $1_a \le (a\rd b)\rd(a\rd b) = 1_{a\rd b}$. Similarly, we have that
\[
1_b (a\rd b) b = (a\rd b) 1_b\cdot b \le (a\rd b)b \le a,
\]
and therefore $1_b \le (a\rd b)\rd (a\rd b) = 1_{a\rd b}$. The proofs of $1_x \le 1_{y\ld x}$ and $1_y \le 1_{y\ld x}$ are similar. The remaining inequalities are a consequence of the previous ones, the monotonicity of the product, and the set of elements below a self-residual being closed under products.
\end{proof}

The partitions $\{\lscr pA : p\in\Idp\m A\}$ and $\{A_p : p\in\Idp\m A\}$ of a balanced residuated semigroup coincide, by Proposition~\ref{prop:equivalent:char:balanced} and Lemma~\ref{lem:self-residuals:positive}. We shall now investigate the case where all these equivalence classes are universes of subalgebras of the residuated semigroup~$\m A$.  We need to impose conditions which ensure that each one of these classes is closed under products and residuals:
\begin{enumerate}[(H1),leftmargin=*,series=H]
    \item\label{H1} $1_x \approx 1_y \implies 1_{xy} \approx 1_x$,
    \item\label{H2} $1_x \approx 1_y \implies 1_{x\rd y} \approx 1_x$,
    \item\label{H3} $1_x \approx 1_y \implies 1_{x\ld y} \approx 1_x$.
\end{enumerate}

We will refer collectively to all these quasiequations as \emph{condition}~(H){\expandafter\def\csname @currentlabel\endcsname{(H)}\label{cond:H}}. But first, we show that they are not completely independent.

\begin{proposition}\label{prop:H2:iff:H3:imply:H1}
In the po-variety of residuated semigroups satisfying $1_x\cdot y \approx y\cdot 1_x$, each one of the quasiequations~\ref{H2} and~\ref{H3} imply~$x\ld x \approx x\rd x$. Moreover, the quasiequations~\ref{H2} and~\ref{H3} are equivalent and they imply~\ref{H1}.
\end{proposition}

\begin{proof}
Let $a$ be an arbitrary element. Since $a1_a = 1_a\cdot a \le a$, we deduce that $1_a \le a\ld a$. Also, the inequality $a(a\ld a)(a\ld a) \le a$ implies $a\ld a \le (a\ld a)\rd (a\ld a) = 1_{a\ld a}$. Notice that $1_{a\ld a} = 1_a$ is a consequence of~\ref{H3}, and therefore~\ref{H3} implies $a\ld a = a\rd a$. Furthermore, $1_{a\ld a}\cdot a(a\ld a) = a1_{a\ld a}\cdot (a\ld a) \le a(a\ld a) \le a$, and therefore $1_{a\ld a} \le (a\rd (a\ld a))\rd a \le (a\rd 1_a)\rd a$, because $1_a \le a\ld a$. Hence,
\[
a\ld a \le 1_{a\ld a} \le (a\rd 1_a)\rd a = a\rd (a1_a) = a\rd (1_a\cdot a) = (a\rd a)\rd 1_a = 1_{a\rd a} = 1_a,
\]
where the last equation is a consequence of~\ref{H2}. Thence, \ref{H2} also implies $a\ld a = a\rd a$.

In order to prove that~\ref{H2} implies~\ref{H3}, suppose that $a,b\in A$ are such that $1_a = 1_b$. Then,
\[
1_{a\ld b} \le 1_{a\ld b\rd b} = 1_{a\ld a\rd a} = 1_{1_a\rd a} = 1_a = 1_b \le 1_{a\ld b},
\]
where the inequalities follow from Lemma~\ref{lem:inequalities:selfresiduals:central}, 
the first equality holds since $b\rd b=1_b=1_a=a\rd a$,
and the third equality follows from~\ref{H2}, since $1_{1_a} = 1_a$. The proof that~\ref{H3} implies~\ref{H2} is analogous.

In order to prove that~\ref{H2} implies~\ref{H1}, suppose again that $a,b\in A$ are such that $1_a = 1_b$. Since $(1_a\rd a)aa \le 1_a\cdot a \le a$, we deduce that $(1_a\rd a)a \le a\rd a = 1_a = 1_b = b\rd b$, and therefore $(1_a\rd a)ab \le b$. Hence, $ab \le (1_a\rd a)\ld b$. Thus,
\begin{align*}
1_{ab} &= ab\rd ab \le (1_a\rd a)\ld b\rd ab = (1_a\rd a)\ld (b\rd b) \rd a = (1_a\rd a)\ld 1_b \rd a \\ 
&= (1_a\rd a)\ld (1_a\rd a) = 1_{1_a\rd a} = 1_a \le 1_{ab},
\end{align*}
where the last equalities follow from~\ref{H2} and the fact that $1_{1_a}=1_a$ and the last inequality follows from Lemma~\ref{lem:inequalities:selfresiduals:central}.
\end{proof}

For every balanced residuated semigroup $\m A$ and every $p\in\Idp\m A$, we can consider the structure $\m A_p = \pair{A_p,\le_p,\cdot_p,\ld_p,\rd\!_p,1_p}$, where the relation $\le_p$ and each of its operations are the restrictions of the corresponding operations of~$\m A$. Condition~\ref{cond:H} ensures that all the operations are total in every $\m A_p$. Not only is the element $1_p$ a global identity of $\m A_p$, but we show that it is the only positive idempotent of $\m A_p$. This property may be expressed equationally: a residuated monoid is called \emph{integrally closed}~\cite{IntegrallyclosedRL} if it satisfies the equation $x \ld x \approx 1$, or equivalently the equation $x \rd x \approx 1$. Both of these equations are equivalent to $1$ being the only positive idempotent, since $a \ld a$ and $a \rd a$ are positive idempotents for each element~$a$, and conversely $a \ld a = a = a \rd a$ for each positive idempotent~$a$. The equivalence between the equations $x \ld x \approx 1$ and $x \rd x \approx 1$ also follows directly from the equations $(x \rd x) \ld (x \rd x) \approx x \rd x$ and $x \ld x \approx (x \ld x) \rd (x \ld x)$ which hold in all residuated monoids.

\begin{proposition}\label{prop:semilattice:sum:decomp:H123}
Every balanced residuated semigroup $\m A$ satisfying condition~\ref{cond:H} decomposes as a family $\{\m A_p : p\in\Idp\m A\}$ of disjoint integrally closed residuated monoids such that the constant-free reduct of $\m A_p$ is a subalgebra of $\m A$ and the monoidal unit of $\m A_p$ is $p$. Finally, $a \ld a = p = a \rd a$ holds for all $a \in A_p$ by the definition of $A_p$.
\end{proposition}

\begin{proof}
The partition $\{A_p : p\in\Idp\m A\}$ is determined by the equivalence relation given by $a\equiv b$ if and only if $1_a = 1_b$. Therefore, the quasiequations~\ref{cond:H} express the conditions that every equivalence class $A_p$ is closed under products and residuals. Accordingly, it is an ordered subalgebra of $\m A$. Moreover, for every $a\in A_p$ we have that $p a = (a \ld a) a = a$, so $p$ is a global identity on $\m A_p$.
\end{proof}

In the context of residuated semigroups rather than residuated monoids, the condition that there is only one positive idempotent and it is the global identity can be expressed by the equations $y (x \ld x) \approx y \approx (x \ld x) y$, or equivalently by the equations $y (x \rd x) \approx y \approx (x \rd x) y$. In that case, the above decomposition is trivial in the sense that it only consists of one algebra $\m A_1$ and $\m A_1 = \widehat{\m A}$, that is, $\m A_1$ is the expansion of $\m A$ by adding its global identity as a constant.

\begin{example}\label{ex:RP:H123}
The left poset $\pair{A,\le}$ of Figure~\ref{fig:H123:but:not:H456} can be equipped with a commutative idempotent multi\-plication, namely, the meet operation of the poset $\pair{A,\sqsubseteq}$ on the right. One can check that this multiplication preserves all joins of $\pair{A,\le}$ and therefore it is residuated. In this way, we obtain a commutative, and therefore balanced, residuated poset $\m A$, where $\Idp\m A = \{1,p,q\}$ and $A_1 = \{1,a\}$, $A_p = \{p\}$, and $A_q = \{q,b,\bot\}$. These sets are indeed closed under residuals, and therefore~$\m A$ satisfies~\ref{cond:H}, by Proposition~\ref{prop:H2:iff:H3:imply:H1}.
\end{example}

\begin{figure}[ht]\centering
\begin{tikzpicture}[baseline=0pt]
\node at (0,-2)[n]{$\le$};
\node(4) at (0,3)[i, label=left:$q$]{};
\node(3) at (0,2)[i, label=left:$p$]{} edge (4);
\node(2) at (1,1)[label=right:$b$]{} edge (3);
\node(1) at (-1,1)[i, label=left:$1$]{} edge (3);
\node(0) at (0,0)[label=right:$a$]{} edge (1) edge (2);
\node(-1) at (0,-1)[label=right:$\bot$]{} edge (0);
\end{tikzpicture}
\qquad\qquad
\begin{tikzpicture}[baseline=0pt]
\node at (0,-2)[n]{$\sle$};
\node(4) at (0,3)[label=right:$1$]{};
\node(3) at (1,2)[label=right:$p$]{} edge (4);
\node(2) at (-1,1.5)[label=left:$a$]{} edge (4);
\node(1) at (1,1)[label=right:$q$]{} edge (3);
\node(0) at (0,0)[label=right:$b$]{} edge (1) edge (2);
\node(-1) at (0,-1)[label=right:$\bot$]{} edge (0);
\end{tikzpicture}

\caption{Residuated poset satisfying condition~\ref{cond:H}.}
\label{fig:H123:but:not:H456}
\end{figure}
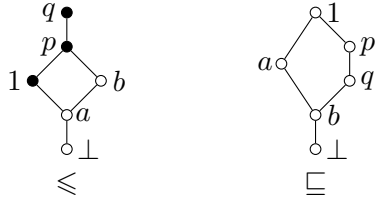

\section{P\l{}onka Sums of Directed Systems of Metamorphisms}\label{plonka}

In the last section we obtained a decomposition result establishing that, under certain minimal conditions, balanced residuated semigroups decompose as families of disjoint subalgebras with exactly one positive idempotent, namely, a global identity. If we want to obtain a ``composition result'', everything points towards the use of P\l{}onka-style constructions, which we explain in the next paragraphs.

Given a join-semilattice of \emph{indices} $\m I = (I, \lor)$, a family of homomorphisms $\Phi = \{\varphi_{pq}\: \m A_p\to \m A_q : p \jle q \text{ in } \m I\}$ between algebras of the same type $\tau$ is said to be a \emph{semilattice directed system} if $\varphi_{pp}$ is the identity on~$\m A_p$ for every index $p$ and $\varphi_{qr}\circ\varphi_{pq} = \varphi_{pr}$ for all indices $p \jle q \jle r$. In other words, a semilattice directed system is a functor $\Phi\: \m I\to \Alg^\tau$, where $\m I$ is understood as a skeletal and thin category with binary coproducts and $\Alg^\tau$ is the category of algebras of type~$\tau$ and homomorphisms.

The \emph{P\l{}onka sum} of a semilattice directed system $\Phi$ is an algebra $\m A$ of the same type, provided that the join-semilattice $\m I$ has a least element $\bot$ if the type contains some constant symbol, defined on the disjoint union of the universes $A = \biguplus_{i\in I} A_i$. For each $n$-ary operation symbol $\sigma$ with $n > 1$ and elements $a_1\in A_{p_1}$, \dots, $a_n\in A_{p_n}$,
\[
\sigma^{\m A}(a_1,\dots,a_n) := \sigma^{\m A_q}(\varphi_{p_1 q}(a_1),\dots,\varphi_{p_n q}(a_n)), \qquad \text{ where } q := p_1\lor \dots\lor p_n.
\]
Moreover, $\omega^{\m A} := \omega^{\m A_\bot}$ for each constant symbol $\omega$. We may more generally admit P\l{}onka sums of algebras of a type which contains a constant symbol even when $\m I$ does not have a least element, with the understanding that the P\l{}onka sum is only an algebra of the constant-free subtype of $\tau$. This construction was first introduced and studied by J. P\l{}onka in~\cite{Plonka67,Plonka68a,Plonka68b} (for more recent expositions see~\cite{PR92} and~\cite{BPP22}). One can readily prove that the P\l{}onka sum of a semilattice directed system is a well-defined algebra of the same type and it satisfies all regular equations---namely, equations $t_1\approx t_2$ such that the variables of $t_1$ and $t_2$ are the same---that hold in all the algebras of the family. 

However, it turns out that this construction is too constrained for our purposes. Indeed, our main result relies on a generalization of the concept of a P\l{}onka sum. In essence, the issue arises because the elements of the system $\Phi$ are homomorphisms, i.e., maps that respect all the operations of the algebra, while we require a more granular approach. We need to treat each operation individually, with its own system of homomorphisms. In addition, we also need to relax the very notion of a map ``respecting'' the algebraic operations. This leads us to the idea of a metamorphism.

Consider two algebras $\m A$ and $\m B$ of the same type $\tau$ and let $(B^A)^*$ denote the set of finite tuples of maps from $A$ to $B$. A \emph{metamorphism} $f\:\m A\meta \m B$ is a map $f\:\tau\to \big(B^A\big)^*$ such that $f\colon \sigma\mapsto f^\sigma = \pair{f^{\sigma 0},\dots,f^{\sigma n}}$ for each $n$-ary symbol $\sigma$ and
\[
f^{\sigma 0}(\sigma^\m A(a_1,\dots,a_n)) = \sigma^\m B(f^{\sigma 1}(a_1),\dots, f^{\sigma n}(a_n)).
\]
In particular, for every constant $\omega$, $f^\omega = \pair{f^{\omega 0}}$ and $f^{\omega 0}(\omega^\m A) = \omega^\m B$. The \emph{composition} of two metamorphisms $f\:\m A\meta\m B$ and $g\:\m B\meta\m C$ is the metamorphism $g\circ f\:\m A\meta\m C$ defined by $(g\circ f)^\sigma := \pair{g^{\sigma 0}\circ f^{\sigma 0},\dots,g^{\sigma n}\circ f^{\sigma n}}$. The \emph{identity} metamorphism on $\m A$ is $\mathrm{id}_\m A\:\m A\meta \m A$ defined by $\mathrm{id}_\m A^\sigma := \pair{\mathrm{id}_A,\dots,\mathrm{id}_A}$. The algebras of type $\tau$ and metamorphisms form a category $\Meta^\tau$, while the algebras of type $\tau$ and homomorphisms form a category $\Alg^\tau$. Each homomorphism $f\:\m A\to\m B$ determines a metamorphism $\widetilde f\:\m A\meta\m B$ such that $\widetilde f^{\sigma 0} = \dots = \widetilde f^{\sigma n} = f$ for each $n$-ary symbol $\sigma$. This assignment defines a faithful functor $\Alg^\tau\to \Meta^\tau$ which is the identity on objects.

A \emph{semilattice directed system of metamorphisms} is a functor $\Xi\: \m I\to\Meta^\tau$ where $\m I = \pair{I,\lor}$ is a join-semilattice. The system $\Xi$ can also be seen as a family $\Xi = \{{\xi_{pq}\:\m A_p\meta\m A_q} : {p\jle q \text{ in } \m I}\}$ such that $\xi_{pp} = \mathrm{id}_{\m A_p}$ and $\xi_{qr}\circ\xi_{pq} = \xi_{pr}$. The \emph{P\l{}onka sum} of a directed system of metamorphisms~$\Xi$ is the algebra $\m A$ of the same type, provided again that $\m I$ has a least element $\bot$ if the type contains constants, whose universe is $A = \biguplus A_p$ and where for every $n$-ary operation $\sigma$ and all $a_1\in A_{p_1}, \dots, a_n\in A_{p_n}$,
\[
\sigma^\m A(a_1,\dots,a_n) := \sigma^{\m A_q}(\xi^{\sigma 1}_{p_1q}(a_1),\dots,\xi^{\sigma n}_{p_nq}(a_n)), \qquad \text{ where } q = p_1\lor\dots\lor p_n.
\]
Moreover, $\omega^\m A := \omega^{\m A_\bot}$ for every constant symbol $\omega$. We may more generally admit a situation where $\m I$ does not have a least element and the type $\tau$ of the algebras $\m A_p$ contains a constant symbol, again with the understanding that the resulting P\l{}onka sum is only an algebra of the constant-free subtype of $\tau$. The algebra $\m A_p$ is called the \emph{fiber} of $\m A$ at $p$ and for that reason we call \emph{fibrant over} $\m I$ any algebra that is the P\l{}onka sum of a directed system of metamorphisms $\m I\to\Meta^\tau$. It follows immediately from its definition that the constant-free reduct of any fiber $\m A_p$ is a subalgebra of the constant-free reduct of $\m A$, and if $\m A$ has constants then $\m A_\bot$ is a subalgebra of~$\m A$. Note also that every directed system of homomorphisms can be viewed as a directed system of metamorphisms, via the composition with the inclusion $\Alg^\tau\to\Meta^\tau$, and the two P\l{}onka sums (as a system of homomorphisms and as a system of metamorphisms) coincide.

P\l{}onka sums can equivalently be understood in terms of so-called partition functions. A \emph{partition function} on an algebra ${\m A}$ is a binary operation $\odot\: A^2\to A$ satisfying the following conditions,\footnote{In the literature, one can find different definitions of a partition function. We are here opting for the definition that appears in~\cite{BPP22} and~\cite{PR92}, which makes use of the minimal number of conditions.} for every $n$-ary operation symbol~$\sigma$, every constant symbol~$\omega$, and all $a, b, c, a_1,\dots, a_n\in A$,
\begin{enumerate}[label=(PF\arabic*), leftmargin=*]
\item\label{PF1} $a\odot a = a$,
\item\label{PF2} $a\odot (b\odot c) = (a\odot b) \odot c $,
\item\label{PF3} $a\odot (b\odot c) = a\odot (c\odot b)$,
\item\label{PF4} $\sigma^{\mathbf{A}}(a_1,\dots,a_n)\odot b = \sigma^{\mathbf{A}}(a_1\odot b,\dots, a_n\odot b)$,
\item\label{PF5} $b\odot \sigma^{\mathbf{A}}(a_1,\dots,a_n) = b\odot a_{1}\odot \dots\odot a_n $,
\item\label{PF6} $b\odot \omega^{\mathbf{A}} = b$. 
\end{enumerate}

P\l{}onka's characterization theorem asserts that an algebra admits a partition function if and only if it is the P\l{}onka sum of a directed system. Not surprisingly, each one of the equalities \ref{PF1}--\ref{PF6} plays a different role in the theorem and, specifically, the equalities that do not involve the operations of the algebra are the ones responsible for the partition of the universe of the algebra and the existence of a directed system connecting these parts.

An algebra $\pair{A,\odot}$ satisfying~\ref{PF1}--\ref{PF3} is known as a \emph{left normal band}. Since these algebras are completely determined by their only operation, we will also call $\odot$ a left normal band on $A$. For every left normal band $\pair{A,\odot}$, we define the binary relations $\jle_\odot$ and $\equiv_\odot$ as follows:
\[
a\jle_\odot b \iff\ b\odot a = b \qquad\text{and}\qquad a\equiv_\odot b \iff a\jle_\odot b \text{ and }\space b\jle_\odot a.
\]

\begin{lemma}\label{lem:PF123:join:semilattice}
For every left normal band $\pair{A,\odot}$ the relation $\jle_\odot$ is a preorder on~$A$ compatible with~$\odot$. Therefore, $\equiv_\odot$ is a congruence of $\pair{A,\odot}$ and the induced relation $\jle'_\odot$ is a partial order on $A/{\equiv}_\odot$. Moreover, $\m I^\odot = \pair{A,\odot}/{\equiv}_\odot$ is a join-semilattice.
\end{lemma}

\begin{proof}
The reflexivity of $\jle_\odot$ follows from~\ref{PF1}. As for the transitivity of $\jle_\odot$, suppose that we have elements $a,b,c\in A$ so that $a\jle_\odot b\jle_\odot c$. Thus,
\[
c\odot a = (c\odot b)\odot a = c\odot (b\odot a) = c\odot b = c,
\]
where the second equality follows from~\ref{PF2}. In order to prove the compatibility of $\jle_\odot$ with $\odot$, suppose that $a\jle_\odot a'$ and $b\jle_\odot b'$.  We will omit in what follows the steps that depend on~\ref{PF2}, that is, the associativity of $\odot$. Therefore,
\[a'\odot b' \odot a\odot b = a'\odot a \odot b'\odot b = a'\odot b',\]
by~\ref{PF3} and the hypotheses. Hence, $a\odot b \jle_\odot a'\odot b'$. Now, to prove that $\pair{A,\odot}/{\equiv}_\odot$ is a join-semilattice, let $a,b\in A$ be arbitrary elements. Then,
\[
a\odot b\odot a = a\odot a\odot b = a\odot b,
\]
by~\ref{PF3} and~\ref{PF1}. We also have that $a\odot b\odot b = a\odot b$, and therefore, we have that $a\jle_\odot a\odot b$ and $b\jle_\odot a\odot b$. Now if $c$ is so that $a\jle_\odot c$ and $b\jle_\odot c$, then $a\odot b\jle_\odot c\odot c = c$, showing that the equivalence class $[a\odot b]$ is the join of the equivalence classes $[a]$ and $[b]$ in $\pair{A/{\equiv}_\odot , \jle'_\odot}$.
\end{proof}

Under the conditions of the previous lemma, let us choose a set $I\subset A$ of representatives of the equivalence classes of $\equiv_\odot$ and let $\m I =\pair{I, \lor}$ be the transport of the structure of $\m I^\odot$ onto $I$. That is, for every $p,q\in I$, $p\lor q$ is the representative of the equivalence class $[p\odot q]$. Thus, $\m I = \pair{I,\lor} \cong \m I^\odot$ and the induced order in $I$ is the restriction of $\jle_\odot$, which we denote in the same way. Let $A_p$ be the equivalence class $[p]$, for every $p\in I$. Given a left normal band $\pair{A,\odot}$, the semilattice $\m I \cong \m I^\odot$ is usually referred to as the \emph{semilattice replica} of $\pair{A,\odot}$ (see \cite{PR92} for details).

\begin{lemma}\label{lem:LNB:induces:DS}
Let $\pair{A,\odot}$ be a left normal band and $\m I = \pair{I,\lor}$ the associated join-semilattice on a set of representatives of $A/{\equiv}_\odot$.
\begin{enumerate}[(1)]
\item For every $p\jle_\odot q$ in $\m I$, the map $\varphi_{pq}\: A_p\to A_q$ given by $\varphi_{pq}(a) := a\odot q$ is well defined.

\item The family $\Phi := \{\varphi_{pq}\: A_p\to A_q : p\jle_\odot q \text{ in }\m I\}$ is a semilattice directed system of maps. 
\end{enumerate}
\end{lemma}

\begin{proof}
For the first claim, notice that $a\odot q \equiv_\odot p\odot q \equiv_\odot q$ for every $a\in A_p$ and $p\jle_\odot q$, and therefore $a\odot q \in A_q$. The map is well defined because left normal bands satisfy the following condition:
\begin{enumerate}[($\dagger$)]
\item\label{*} if $b\equiv_\odot c$, then $a\odot b = a\odot c$.
\end{enumerate}
Indeed, if $b\equiv_\odot c$, then $c\odot b = c$ and $b\odot c = b$, and therefore
\[
a\odot b = a\odot b\odot c = a\odot c\odot b = a\odot c.
\]

Finally, if $a\in A_p$, then in particular $a\jle_\odot p$ and therefore $a\odot p = a$, whence we deduce that $\varphi_{pp}$ is the identity map of $A_p$; and if $p\jle_\odot q\jle_\odot r$, then $r\odot q = r$ and $\varphi_{qr}(\varphi_{pq}(a)) = a\odot q\odot r = a\odot r\odot q = a\odot r = \varphi_{pr}(a)$, for every $a\in A_p$.
\end{proof}

\begin{lemma}\label{lem:DS:induces:LNB}
Let $\Phi = \{\varphi_{pq}\: A_p\to A_q : p\jle q \text{ in }\m I\}$ be a directed system of maps, where $\m I$ is a join-semilattice. The binary operation $\odot$ on $A = \biguplus_{p\in I} A_p$ defined by
\[
a\odot b := \varphi_{ps}(a),\qquad\text{where } a\in A_p,\ b\in A_q,\ s := p\lor q,
\]
is a left normal band, $a\jle_\odot b$ if and only if $p\jle q$, and $\m I^\odot \cong \m I$.
\end{lemma}

\begin{proof}
Consider arbitrary $p,q,r\in A$ and $a\in A_p$, $b\in A_q$, and $c\in A_r$. Let $s := p\lor q$, $t := q\lor r$, and $v := p\lor t = s\lor r$.
\begin{enumerate}[(PF1), leftmargin=*]
\item $a\odot a = \varphi_{pp}(a) = a$.
\item $\begin{aligned}[t]
    a\odot (b\odot c) &= a \odot \varphi_{qt}(b) = \varphi_{pv}(a) = \varphi_{sv}(\varphi_{ps}(a)) = \varphi_{ps}(a)\odot c = (a\odot b)\odot c.
\end{aligned}$
\item $a\odot (b\odot c) = a\odot\varphi_{qt}(b) = \varphi_{pv}(a) = a\odot\varphi_{rt}(c) = a\odot (c\odot b)$.
\end{enumerate}

Finally, consider the projection map $\pi\: A\to I$ that assigns to every element of the disjoint union $A = \biguplus_{p\in I} A_p$ its index. Under the same hypothesis as above, we have
\[
a \jle_\odot b \iff b\odot a = b \iff \varphi_{qs}(b) = b \iff q=s \iff p\jle q \iff \pi(a)\jle\pi(b).
\]
Therefore, since $\pi$ is surjective, $\m I^\odot\cong \m I$.
\end{proof}

In what follows, we will consider several bands with the same underlying  universe. We say that two left normal bands $\odot$ and $\otimes$ on a set $A$ are \emph{homotactic} if ${\jle}_\odot = {\jle}_\otimes$ and therefore ${\equiv}_\odot = {\equiv}_\otimes$. In particular, if $\odot$ and $\otimes$ are homotactic, then $\m I^\odot = \m I^\otimes$. In other words, the two left normal bands are homotactic when they have the same semilattice replica. The corresponding generalization of a partition function  for an algebra $\m A$ is that of a \emph{partition system}, which is an assignment $\tau\to O^*$, $\sigma\mapsto\pair{\odot^{\sigma}_0,\dots,\odot^{\sigma}_n}$, where $O$ is a set of homotactic left normal bands on $A$, such that for all $a_1,\dots,a_n,b\in A$,
\begin{enumerate}[align=left, left=0pt .. \parindent]
\myitem[(PF4$^\sigma$)]\label{PF4s} $\sigma^\m A(a_1,\dots,a_n)\odot^\sigma_0 b = \sigma^\m A(a_1\odot^{\sigma}_1 b,\dots,a_n\odot^{\sigma}_n b)$,\quad if $n\geqslant 1$,

\myitem[(PF5$^\sigma$)]\label{PF5s} $b\odot^\sigma_0 \sigma^{\mathbf{A}}(a_1,\dots,a_n) = b\odot^\sigma_0 a_{1}\odot^\sigma_0 \dots\odot^\sigma_0 a_n $,
\end{enumerate}

In particular, for a constant symbol $\omega$, the property~\ref{PF5s} can be written as follows:

\begin{enumerate}[align=left, left=0pt .. \parindent]
\myitem[(PF5$^\omega$)]\label{PF5w} $b\odot^\omega_0 \omega^\m A = b$.
\end{enumerate}

Notice that the conditions~\ref{PF5s} and~\ref{PF5w} are exactly~\ref{PF5} and~\ref{PF6} for the left normal bands $\odot^\sigma_0$ and $\odot^\omega_0$, respectively. A partition function can be equated to a partition system such that $O$ has exactly one left normal band.

\begin{lemma}\label{lem:A:PF5:closure:under:sigma}
Let $\m A$ be an algebra with a left normal band $\odot$ on $A$.
\begin{enumerate}[(1)]
\item If $\sigma^\m A$ is an $n$-ary operation which satisfies~\ref{PF5}, then the following hold.

\begin{enumerate}
\item For all $a_1,\dots, a_n\in A$, $\sigma^\m A(a_1,\dots,a_n) \equiv_\odot a_1\odot\dots\odot a_n$.

\item Every equivalence class $A_p$ is closed under $\sigma^\m A$.
\end{enumerate}    

\item If $\omega^\m A$ is a constant which satisfies~\ref{PF6}, then the class $[\omega^\m A]$ is the least element of $\m I^\odot$.
\end{enumerate}
\end{lemma}

\begin{proof}
For the first claim, notice that
\begin{align*}
\sigma^\m A(a_1,\dots,a_n) \odot a_i &= \sigma^\m A(a_1,\dots,a_n) \odot \sigma^\m A(a_1,\dots,a_n) \odot a_i \\
    &= \sigma^\m A(a_1,\dots,a_n) \odot a_1 \odot \dots \odot a_n\odot a_i\\
    &= \sigma^\m A(a_1,\dots,a_n) \odot a_1 \odot \dots \odot a_n\\
    &= \sigma^\m A(a_1,\dots,a_n) \odot \sigma^\m A(a_1,\dots,a_n)\\
    &= \sigma^\m A(a_1,\dots,a_n).
\end{align*}
The first and the last equalities follow from~\ref{PF1} and the second and fourth from~\ref{PF5}. The third equality is a consequence of a number of applications of~\ref{PF2} and~\ref{PF3} and one application of~\ref{PF1}. Therefore, $a_i\jle_\odot  \sigma^\m A(a_1,\dots,a_n)$ and thus $a_1\odot\dots\odot a_n \jle_\odot \sigma^\m A(a_1,\dots,a_n)$. For the reverse inequality, take $q\in I$ so that $q \equiv_\odot a_1\odot\dots\odot a_n$. Then,
\[
q\odot\sigma^\m A(a_1,\dots,a_n) = q\odot a_1\odot \dots\odot a_n \equiv_\odot q\odot q = q,
\]
and hence $\sigma^\m A(a_1,\dots,a_n) \jle_\odot q \equiv_\odot a_1\odot\dots\odot a_n$. The second part follows immediately, because if $a_1,\dots,a_n\in A_p$, then 
\[
\sigma^\m A(a_1,\dots,a_n) \equiv_\odot a_1\odot \dots\odot a_n \equiv_\odot p\odot \dots\odot p = p,
\]
that is, $\sigma^\m A(a_1,\dots,a_n) \in A_p$. Finally, if $a\odot \omega^\m A = a$ holds for all $a\in A$, then $\omega^\m A \jle_\odot a$, whence the last claim follows.   
\end{proof}

The following two theorems express the relation between P\l{}onka sums of systems of metamorphisms and partitions system, which is in essence that an algebra is fibrant over $\m I$ if and only if it admits a partition system with induced join-semilattice isomorphic to $\m I$.

\begin{theorem}\label{thm:Plonka:decomposition:meta}
Let $\m A$ be an algebra with a partition system, $\equiv$ and $\m I = \pair{I,\lor}$ their induced equivalence relation and corresponding join-semilattice on a set of representatives, and $\jle$ the order of $\m I$.
\begin{enumerate}[(1), leftmargin=*]
\item Every equivalence class $A_p$ of $\equiv$ is the universe of an algebra $\m A_p$ of the same type whose constant-free reduct is a subalgebra of the constant-free reduct of~$\m A$, and $\omega^{\m A_p} = \omega^{\m A}\odot^\omega_0 p$ for every constant symbol $\omega$.

\item For all $p\jle q$ in $\m I$, there is a metamorphism $\xi_{pq}\:\m A_p\meta\m A_q$ defined by
\[
\xi^{\sigma i}_{pq}(a) := a\odot^\sigma_i q, \quad \text{for every $n$-ary symbol $\sigma$ and $i\le n$}.
\]

\item The family $\Xi := \{\xi_{pq}\:\m A_p\meta\m A_q : p\jle q \text{ in }\m I\}$ is a semilattice directed system of metamorphisms and $\m A$ is its P\l{}onka sum.
\end{enumerate}
\end{theorem}

\begin{proof}
\begin{enumerate}[(1),leftmargin=*]
\item Since every $n$-ary operation symbol $\sigma$ satisfies~\ref{PF5s}, i.e., \ref{PF5} for $\sigma$ and $\odot^\sigma_0$, every equivalence class $A_p$ is closed under $\sigma^\m A$, by Lemma~\ref{lem:A:PF5:closure:under:sigma}.

\item First, notice that every map $\xi^{\sigma i}_{pq}$ for $p\jle q$ is well defined by Lemma~\ref{lem:LNB:induces:DS}. By definition of $\xi^\sigma_{pq}$, part~(1), and~\ref{PF4s} we have
\begin{align*}
\xi^{\sigma 0}_{pq}(\sigma^{\m A_p}(a_1,\dots, a_n))  &= \sigma^\m A(a_1,\dots, a_n)\odot^\sigma_0 q 
= \sigma^\m A(a_1\odot^{\sigma}_1 q,\dots, a_n\odot^{\sigma}_n q) \\
&= \sigma^{\m A_q}(\xi^{\sigma 1}_{pq}(a_1),\dots, \xi^{\sigma n}_{pq}(a_n)).
\end{align*}
And if $\omega$ is a constant symbol, then
\[
\xi^{\omega 0}_{pq}(\omega^{\m A_p}) = \xi^{\omega 0}_{pq}(\omega^{\m A}\odot^\omega_0 p) 
= \omega^{\m A}\odot^\omega_0 p\odot^\omega_0 q
= \omega^{\m A}\odot^\omega_0 q\odot^\omega_0 p
= \omega^{\m A}\odot^\omega_0 q = \omega^{\m A_q},
\]
since $p\jle q$ in $\m I$. Thus, $\xi_{pq}\:\m A_p\meta\m A_q$ is a metamorphism.

\item Again by virtue of Lemma~\ref{lem:LNB:induces:DS}, if $p\jle q\jle r$ in $\m I$, then for every $n$-ary operation symbol $\sigma$ and all $i\le n$,
\[
(\xi_{qr}\circ\xi_{pq})^{\sigma i} = \xi_{qr}^{\sigma i}\circ\xi_{pq}^{\sigma i} = \xi^{\sigma i}_{pr}\quad\text{and}\quad \xi_{pp}^{\sigma i} = \mathrm{id}_{A_p},
\]
and hence $\xi_{qr}\circ\xi_{pq} = \xi_{pr}$ and $\xi_{pp} = \widetilde{\mathrm{id}}_{A_p}$. Finally, by Lemma~\ref{lem:A:PF5:closure:under:sigma},
\[
\sigma^\m A(a_1,\dots,a_n)\equiv a_1\odot^\sigma_0 \dots \odot^\sigma_0 a_n \equiv p_1\lor\dots\lor p_n = q,
\]
and by~\ref{PF4s} and part~(1),
\begin{align*}
\sigma^\m A(a_1,\dots,a_n) &= \sigma^\m A(a_1,\dots,a_n)\odot^\sigma_0 q 
= \sigma^\m A(a_1\odot^{\sigma}_1 q,\dots,a_n\odot^{\sigma}_n q)\\ 
&= \sigma^{\m A_q}(\xi_{p_1q}^{\sigma 1}(a_1), \dots, \xi_{p_nq}^{\sigma n}(a_n)).
\end{align*}
And if $\omega$ is a constant symbol, then $\m I$ contains a least element $\bot$ and $\omega^\m A\in A_\bot$, by~\ref{PF5w}. Hence, $\omega^{\m A_\bot} = \omega^\m A\odot^\omega_0\bot = \omega^\m A$.\QED
\end{enumerate}\let\QED\relax
\end{proof}

\begin{theorem}\label{thm:Plonka:composition:meta}
Every semilattice directed system of metamorphisms is induced by a partition system for its P\l{}onka sum.
\end{theorem}

\begin{proof}
Let $\Xi = \{\xi_{pq}\:\m A_p\meta\m A_q : p\jle q \text{ in }\m I\}$ be a semilattice directed system of metamorphisms and $\m A$ its P\l{}onka sum. For every $n$-ary operation or constant symbol $\sigma$ and every natural $i\le n$, we define the operation $\odot^\sigma_i\: A^2\to A$ as follows: for every $a\in A_p$ and $b\in A_q$
\[
a\odot^\sigma_i b = \xi^{\sigma i}_{ps}(a), \qquad \text{where } s:=p\lor q.
\]
We will show that the assignment $\sigma\mapsto\pair{\odot^\sigma_0,\dots,\odot^\sigma_n}$ is a partition system for $\m A$ and that its induced directed system of metamorphisms is $\Xi$.

By Lemma~\ref{lem:DS:induces:LNB}, every $\odot^\sigma_i$ is a left normal band on $A$, all of them are homotactic and the induced common partition is $\{A_p : p\in I\}$. As for the rest of the properties, suppose that $b\in A_q$, $a_i\in A_{p_i}$ and $s_i = p_i\lor q$, for $i=1,\dots, n$,  $w = p_1\lor\dots\lor p_n$, $x = w\lor q$, and $t_i = p_1\lor\dots\lor p_i\lor q$, for all $i=1,\dots,n$, and therefore $t_n = x$.
\begin{enumerate}[align=left, left=0pt .. \parindent]
\item[(PF4$^\sigma$)] $\begin{aligned}[t]
    \sigma^\m A(a_1,\dots,a_n)\odot^\sigma_0 b &= \sigma^{\m A_w}(\xi^{\sigma 1}_{{p_1}w}(a_1),\dots, \xi^{\sigma n}_{{p_n}w}(a_n))\odot^\sigma_0 b \\
    &= \xi^{\sigma 0}_{wx}(\sigma^{\m A_w}(\xi^{\sigma 1}_{{p_1}w}(a_1),\dots, \xi^{\sigma n}_{{p_n}w}(a_n))) \\
    &= \sigma^{\m A_w}(\xi^{\sigma 1}_{wx}\xi^{\sigma 1}_{{p_1}w}(a_1),\dots, \xi^{\sigma n}_{wx}\xi^{\sigma n}_{{p_n}w}(a_n)) \\
    &= \sigma^{\m A_w}(\xi^{\sigma 1}_{{s_1}x}\xi^{\sigma 1}_{{p_1}{s_1}}(a_1),\dots, \xi^{\sigma n}_{{s_n}x}\xi^{\sigma n}_{{p_n}{s_n}}(a_n)) \\
    &= \sigma^{\m A_w}(\xi^{\sigma 1}_{{s_1}x}(a_1\odot^\sigma_1 b),\dots, \xi^{\sigma n}_{{s_n}x}(a_n\odot^\sigma_n b)) \\
    &= \sigma^\m A(a_1\odot^\sigma_1 b,\dots, a_n\odot^\sigma_n b).
\end{aligned}$

\item[(PF5$^\sigma$)] $\begin{aligned}[t]
    b\odot^\sigma_0 \sigma^\m A(a_1,\dots, a_n) &= b\odot^\sigma_0 \sigma^{\m A_w}(\xi^{\sigma 1}_{{p_1}w}(a_1),\dots, \xi^{\sigma 1}_{{p_1}w}(a_n)) = \xi^{\sigma 0}_{qx}(b) \\
    &=\xi^{\sigma 0}_{t_{n-1}t_n}\cdots\, \xi^{\sigma 0}_{t_1t_2}\xi^{\sigma 0}_{qt_1}(b) \\
    &=\xi^{\sigma 0}_{t_{n-1}t_n}\cdots\, \xi^{\sigma 0}_{t_1t_2}(b\odot^\sigma_0  a_1) \\
    &=\xi^{\sigma 0}_{t_{n-1}t_n}(b\odot^\sigma_0  a_1\odot^\sigma_0 \dots \odot^\sigma_0 a_{n-1}) \\
    &= b\odot^\sigma_0  a_1\odot^\sigma_0 \dots \odot^\sigma_0 a_n.
\end{aligned}$

\item[(PF6$^\omega$)] $b\odot^\omega_0 \omega^\m A = \xi^{\omega 0}_{qq}(b) = b$, for every constant $\omega$, since $\omega^\m A \in A_\bot$ and $\bot \jle q$.
\end{enumerate}

Thus we have seen that the assignment $\sigma\mapsto\pair{\odot^\sigma_0,\dots,\odot^\sigma_n}$ is a partition system for $\m A$. Let $Z=\{\zeta_{pq}\:\m A_p\to\m A_q : p\jle q \text{ in }\m I\}$ be its induced directed system of metamorphisms. Thus, for all $p\jle q$ in $\m I$ and $a\in A_p$, and every $n$-ary symbol~$\sigma$ and $i=1,\dots,n$, we have that $\zeta^{\sigma i}_{pq}(a) = a\odot^\sigma_i q = \xi^{\sigma i}_{pq}(a)$. That is, $Z = \Xi$.
\end{proof}

The main reason for introducing this notion of P\l{}onka sum, as we have already mentioned, is to establish an appropriate framework to study the compositions of residuated semigroups. This construction in fact goes beyond the theory of residuated semigroups and it would be interesting to explore it further in subsequent work. Nonetheless, to conclude this section, we present a result in this direction that shows how an algebraic property can be transferred from a fibrant algebra to its fibers and vice versa. Recall that an algebra $\m A$ is \emph{locally finite} if the subalgebra generated by any finite set of elements of $\m A$ is finite.

\begin{theorem}
A fibrant algebra of finite type is locally finite if and only if all its fibers are locally finite.
\end{theorem}

\begin{proof}
Let $\Xi = \{\xi_{pq}\:\m A_p\meta\m A_q : p\jle q \text{ in }\m I\}$ be a semilattice directed system of metamorphisms between algebras of finite type and $\m A$ its P\l{}onka sum. If $\m A$ is locally finite, then every fiber is also locally finite, because fibers are subalgebras of $\m A$. Suppose now that all the fibers of $\m A$ are locally finite. We will just consider the case in which $\m I$ contains a least element $\bot$. If it does not (in which case the type does not contain any constant symbols), the proof would be analogous but simpler. For every finite set $X\subset A$, consider the set
\[
J_X = \{p_1\lor\dots\lor p_n : \text{ there are } a_1,\dots, a_n \in X,\ a_i\in A_{p_i}\},
\]
which is also finite, closed under joins, and contains $\bot$ (for $n=0$). We will prove that there is a finite subalgebra $\m B\le \m A$ such that $X\subset B$, $J_B = J_X$, and $\xi^{\sigma i}_{rs}(a)\in B$ for every $n$-ary symbol $\sigma$, all $1\le i\le n$, all $r,s\in J_X$ such that $r\jle s$, and every $a\in B\cap A_r$, by induction on the size of $J_X$.

If $J_X = \{\bot\}$, then $X\subset A_\bot$. Consider the subalgebra $\m B\le\m A_\bot$ generated by~$X$. Since $\m A_\bot$ is a locally finite subalgebra of $\m A$, we deduce that $\m B\le\m A$, $\m B$ is finite, $X\subset B$, $J_B = \{\bot\} = J_X$. The last part of the statement is trivially satisfied.

In the remainder of the proof, we deal with the case where $J_X\neq\{\bot\}$. Let $p$ be a minimal element of $J_X\setminus\{\bot\}$. Consider the subalgebra $\m C\le \m A_\bot$ generated by $X\cap A_\bot$, which is finite, because $X\cap A_\bot$ is finite and $\m A_\bot$ is locally finite. Define the set
\[
Y = (X\cap A_p)\cup\{\xi^{\sigma i}_{\bot p}(a) : a\in C, \text{ $n$-ary $\sigma$},\ 1\le i\le n\} \subset A_p,
\]
and consider the subalgebra $\m D \le \m A_p$ generated by $Y$, which is finite because $Y$ is finite and $\m A_p$ is locally finite. Define the set
\begin{align*}
Z = \big(X\setminus A_p\big) 
&\cup \big\{\xi^{\sigma i}_{\bot q}(a) : a\in C, \text{ $n$-ary $\sigma$},\ 1\le i\le n,\ q\in J_X\setminus\{p\}\big\} \\
&\cup \big\{\xi^{\sigma i}_{pq}(a) : a\in D, \text{ $n$-ary $\sigma$},\ 1\le i\le n,\ p\jle q\in J_X\setminus\{p\}\big\}.    
\end{align*}
Notice that $Z$ is finite and $J_Z = J_X\setminus\{p\}$ because of our choice of $p$. Hence, by the inductive hypothesis, there is an algebra $\m B\le\m A$ such that $Z\subset B$, $J_B = J_Z$, and $\xi^{\sigma i}_{rs}(a)\in B$ for every $n$-ary symbol $\sigma$, all $1\le i\le n$, all $r,s\in J_Z$ such that $r\jle s$, and $a\in B\cap A_r$.

Consider the finite set $E = (B\setminus A_\bot)\cup C\cup D$, which is actually a disjoint union. Obviously, $X\subset E$ and $J_E = J_X$. Suppose that $\sigma$ is an $n$-ary symbol and $1\le i\le n$; consider $r,s\in J_X$ such that $r\jle s$ and $a\in E\cap A_r$. One of the following cases would hold:
\begin{enumerate}[(a)]
    \item If $r=s=\bot$, then $\xi^{\sigma i}_{rs}(a) = \xi^{\sigma i}_{\bot\bot}(a) = a\in C$.

    \item If $r=s=p$, then $\xi^{\sigma i}_{rs}(a) = \xi^{\sigma i}_{pp}(a) = a\in D$.
    
    \item If $r=\bot$ and $s=p$, then $a\in C$ and $\xi^{\sigma i}_{rs}(a) = \xi^{\sigma i}_{\bot p}(a) \in Y\subset D$.
    
    \item If $r=\bot$ and $s\neq p$, then $a\in C$ and $\xi^{\sigma i}_{rs}(a) = \xi^{\sigma i}_{\bot s}(a)\in Z\subset B$. 

    \item If $r=p$ and $s\neq p$, then $a\in D$ and $\xi^{\sigma i}_{rs}(a) = \xi^{\sigma i}_{ps}(a) \in Z\subset B$. 
    
    \item Otherwise, $r,s\in J_Z\setminus\{\bot\}$ and then $a\in B\cap A_r$ and $\xi^{\sigma i}_{rs}(a)\in B$.
\end{enumerate}
In any event, $\xi^{\sigma i}_{rs}(a)\in E$.

In order to check that $E$ is the universe of a subalgebra of $\m A$, notice first that it contains all the constants of $\m A$ because $\m C$ is a subalgebra of~$\m A$. Let's check that $E$ is also closed under every other $n$-ary operation $\sigma$. Suppose that $a_1,\dots,a_n\in E$, where $a_i\in A_{p_i}$ for all $1\le i\le n$. In particular, $p_1,\dots,p_n\in J_X$ and hence $q = p_1\lor\dots\lor p_n\in J_X$.

\begin{itemize}[$\circ$]
    \item If $q=\bot$, then $p_1 = \dots = p_n = \bot$ and hence $a_1,\dots,a_n\in C$ and $\sigma^\m A(a_1,\dots,a_n)\in C$, because $\m C$ is a subalgebra of $\m A$.

    \item If $q=p$, then $a_1,\dots,a_n\in C\cup D$. If $a_i\in C$, then $p_i = \bot$ and $\xi^{\sigma i}_{p_ip}(a_i) = \xi^{\sigma i}_{\bot p}(a_i)\in Y\subset D$; and if $a_i\in D$, then $p_i = p$ and $\xi^{\sigma i}_{p_ip}(a_i) = \xi^{\sigma i}_{pp}(a_i) = a_i\in D$. Hence,
    \[
    \sigma^\m A(a_1,\dots,a_n) = \sigma^{\m A_p}(\xi^{\sigma 1}_{p_1p}(a_1),\dots, \xi^{\sigma n}_{p_np}(a_n))\in D,
    \]
    because $\m D$ is a subalgebra of $\m A_p$.

    \item Otherwise, $q\in J_Z$. By parts~(d), (e), and~(f), we have that $\xi^{\sigma i}_{p_iq}(a_i)\in B$, and therefore
    \begin{align*}
    \sigma^\m A(a_1,\dots,a_n) &= \sigma^{\m A_q}(\xi^{\sigma 1}_{p_1q}(a_1),\dots, \xi^{\sigma n}_{p_nq}(a_n)) \\
    &= \sigma^{\m A}(\xi^{\sigma 1}_{p_1q}(a_1),\dots, \xi^{\sigma n}_{p_nq}(a_n))\in B,
    \end{align*}
    because $\m B$ is a subalgebra of $\m A$.\QED
\end{itemize}\let\QED\relax
\end{proof}

\section{Sums of Posets and Partially Ordered Algebras}\label{sec:sums:of:posets}

The theory of P\l{}onka sums has also been generalized to encompass the case of (sums of) arbitrary relational systems~\cite{Plonka73}. However, such an approach does not seem to be very fruitful to treat partially ordered sets. In this section, we introduce and investigate the correct notions of decomposition and reconstruction of partially ordered set that will serve our purposes. More concretely, given a family $\{\m A_p : p\in I\}$ of disjoint posets, we would like to define an order $\le$ on the disjoint union $A = \biguplus A_p$ of their universes \emph{extending} each one of the partial orders, meaning that ${\le}\cap A_p^2 = {\le}_p$, for every $p\in I$. Obviously, we could just take ${\le} = \bigcup_{p \in I}{\le}_p$, but this would be insufficient for our needs, since the families of posets in which we are interested are actually connected via some maps, and the order $\le$ should be compatible, in some sense, with these maps. 

Consider the category $\Pos^\parallel$ of posets and pairs of parallel monotone maps $\pair{\psi,\varphi}\:\m A\to \m A'$. A~\emph{directed system of parallel monotone maps} is a functor $\Xi\: \m I\to\Pos^\parallel$ where $\m I = \pair{I,\lor}$ is a join-semilattice. As in the previous section, we can also equate a directed system to a coherent family $\{\pair{\psi_{pq},\varphi_{pq}}\: {\m A_p\to \m A_q} : p\jle q \text{ in }\m I\}$. Given such a directed system, we define the relation $\le$ on $A := \biguplus A_p$ as follows: for all $p,q\in I$, $a\in A_p$, and $b\in A_q$,
\begin{equation}\label{eq:def:sum:order}
a\le b \quad\iff\quad \varphi_{ps}(a) \le_s \psi_{qs}(b),\qquad \text{where } s:=p\lor q.    
\end{equation}
We call $\m A = \pair{A,\le}$ the \emph{sum} of the directed system and $\m A_p$ the \emph{fiber} of $\m A$ at $p\in I$. Now, this relation does not have to be an order in general, but it will be one if the following two conditions are satisfied. 

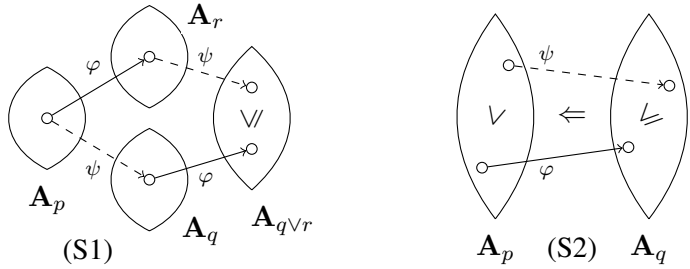
\begin{center}
\begin{tikzpicture}[scale=3]
\draw (0,1.5)..controls(-.5,1.75)and(-.5,2.25)..(0,2.5)
..controls(.5,2.25)and(.5,1.75)..(0,1.5);
\node at (0,1.2)[n]{$\mathbf A_p$};
\draw (1,2.1)..controls(.5,2.35)and(.5,2.85)..(1,3.1)
..controls(1.5,2.85)and(1.5,2.35)..(1,2.1);
\node at (1.55,3)[n]{$\mathbf A_r$};
\node at (0.45,2.5)[n]{\scriptsize $\varphi$};
\node at (0.45,1.5)[n]{\scriptsize $\psi$};
\node at (1.55,2.6)[n]{\scriptsize $\psi$};
\node at (1.55,1.4)[n]{\scriptsize $\varphi$};
\node at (0.4,0.7)[n]{(S1)};
\node(2) at (2,2.3){};
\node(3) at (2,1.7){};
\node(0) at (1,2.6){} edge[dashed,->](2);
\node(1) at (1,1.4){} edge[->](3);
\node at (0,2){} edge[->](0) edge[dashed,->](1);
\draw (1,0.9)..controls(.5,1.15)and(.5,1.65)..(1,1.9)
..controls(1.5,1.65)and(1.5,1.15)..(1,0.9);
\node at (1.5,.9)[n]{$\mathbf A_q$};
\draw (2,1.3)..controls(1.5,1.75)and(1.5,2.25)..(2,2.7)
..controls(2.5,2.25)and(2.5,1.75)..(2,1.3);
\node at (2.3,1)[n]{$\mathbf A_{q{\vee}r}$};
\node at (2,2)[n]{\rotatebox[origin=c]{90}{$\le$}};
\end{tikzpicture}
\qquad\qquad 
\begin{tikzpicture}[scale=3]
\draw (0,1)..controls(-.5,1.66)and(-.5,2.33)..(0,3)
..controls(.5,2.33)and(.5,1.66)..(0,1);
\node at (0,0.7)[n]{$\mathbf A_p$};
\node at (0.5,2.6)[n]{\scriptsize $\psi$};
\node at (0.5,1.45)[n]{\scriptsize $\varphi$};
\node at (0.75,0.7)[n]{(S2)};
\node(0) at (1.3,1.7){};
\node(1) at (1.7,2.3){};
\node at (-.13,1.5){} edge[->](0);
\node at (.13,2.5){} edge[dashed,->](1);
\draw (1.5,1)..controls(1,1.66)and(1,2.33)..(1.5,3)
..controls(2,2.33)and(2,1.66)..(1.5,1);
\node at (0.75,2)[n]{$\Leftarrow$};
\node at (1.5,0.7)[n]{$\mathbf A_q$};
\node at (1.5,2)[n]{\rotatebox[origin=c]{45}{$\le$}};
\node at (0,2)[n]{\rotatebox[origin=c]{75}{$<$}};
\end{tikzpicture}
\end{center}

\begin{enumerate}[(S1),leftmargin=*,series=S]
\item\label{S1} If $p \jl q,r$ and $t = q\lor r$, then $\varphi_{qt}\psi_{pq} \le_t \psi_{rt}\varphi_{pr}$ pointwise.

\item\label{S2} If $p \jl q$ and $a,b\in A_p$ and $\varphi_{pq} (a) \le_q \psi_{pq} (b)$, then $a <_p b$.
\end{enumerate}

\begin{theorem}\label{thm:char:sums:of:posets:over:directed:systems}
Given a directed systems of monotone maps $\Xi$, the relation $\le$ defined by~\eqref{eq:def:sum:order} is a partial order extending the order of each poset if and only if $\Xi$ satisfies~\ref{S1} and~\ref{S2}.
\end{theorem}

\begin{proof} 
Let's first prove that conditions~\ref{S1} and~\ref{S2} are necessary so that definition~\eqref{eq:def:sum:order} determines a partial order on $A = \biguplus A_p$ extending the order of every~$\m A_p$. If $p \jl q$ and $a\in A_p$, then $q = p\lor q$ and $\varphi_{qq}(\psi_{pq}(a) )\le_q \psi_{pq}(a)$, and therefore $\psi_{pq}(a)\le a$. Analogously, if $p \jl r$, then $a\le\varphi_{pr}(a)$, and therefore $\psi_{pq}(a)\le\varphi_{pr}(a)$, by the transitivity of $\le$. By virtue of definition~\eqref{eq:def:sum:order}, $\varphi_{qt}(\psi_{pq}(a)) \le_t \psi_{rt}(\varphi_{pr}(a))$, where $t = q\lor r$. Hence, \ref{S1} holds. Moreover, if $p \jl q$ and $a,b\in A_p$ are such that $\varphi_{pq}(a)\le_q\psi_{pq}(b)$, then
\[
a \le \varphi_{pq}(a) \le \psi_{pq} (b) \le b,
\]
as we have shown above, and by transitivity we obtain that $a\le b$, and in particular $a\le_p b$. But if $a = b$, then it follows that $\varphi_{pq}(a) = a$, which is impossible because $p \jl q$. Therefore, $a <_p b$, which proves~\ref{S2}.

Suppose now that $\Xi$ satisfies~\ref{S1} and~\ref{S2}. We first prove that the following three properties are also satisfied.

\begin{enumerate}[resume*=S]
\item\label{S3} If $p \jle q,r$ and $t = q\lor r$, then $\varphi_{qt}\psi_{pq} \le_t \psi_{rt}\varphi_{pr}$ pointwise.
\item\label{S4} If $p\jle q$, then $\psi_{pq} \le_q \varphi_{pq}$ pointwise. 
\item\label{S5} If $p \jl q$, then $\psi_{pq} <_q \varphi_{pq}$ pointwise.
\end{enumerate}

For the first one, notice that if $p = q$, then $t = r$ and the inequation becomes $\varphi_{pr}\psi_{pp} \le_r \psi_{rr}\varphi_{pr}$, which is true because $\psi_{pp} = \mathrm{id}_{A_p}$ and $\psi_{rr} = \mathrm{id}_{A_r}$. The proof that the inequation is true when $p = r$ is analogous. Concerning the second property, notice that taking $q=r$ in~\ref{S3} it follows that $q=t=r$ and $\varphi_{qq}\psi_{pq}\le_q \psi_{qq}\varphi_{pq}$, that is, $\psi_{pq} \le_q \varphi_{pq}$. As for the third property, notice that if $p \jl q$ and $\psi_{pq}(a) = \varphi_{pq}(a)$ for some $a\in A_p$, then it follows from~\ref{S2} that $a <_p a$, which is impossible. It follows from~\ref{S4} that $\psi_{pq} <_q \varphi_{pq}$.

The fact that the restriction of $\le$ to $A_p$ is $\le_p$ for every $p\in I$ follows immediately from definition~\eqref{eq:def:sum:order} and the fact that that $\varphi_{pp} = \psi_{pp} = \mathrm{id}_{A_p}$. Indeed, if $a,b\in A_p$ are such that $a\le b$, then $p = p\lor p$ and $a = \varphi_{pp}(a) \le_p \psi_{pp}(b) = b$. And reciprocally, if $a\le_p b$, then $\varphi_{pp}(a) = a \le_p b = \psi_{pp}(b)$, and therefore $a\le b$. In particular, $\le$~is reflexive.

In order to prove antisymmetry, suppose that $a\in A_p$ and $b\in A_q$ are such that $a\le b$ and $b\le a$ and consider $s = p\lor q$. If $p \jl s$, then
\[
\varphi_{ps}(a) \le_s \psi_{qs} (b) \le_s \varphi_{qs} (b) \le_s \psi_{ps} (a) <_s \varphi_{ps} (a),
\]
by the definition of $\le$ and conditions~\ref{S4} and~\ref{S5}, which is impossible, and therefore $p = s$. For the same reason $q = s$ and therefore $p = q$, $a\le_p b$, and $b\le_p a$, whence $a=b$.

As for transitivity, suppose that $a\in A_p$, $b\in A_q$, and $c\in A_r$ are such that $a\le b\le c$. Consider $s = p\lor q$, $t = q\lor r$, $u = p\lor r$, $v = s\lor t$, and notice that $u = p\lor r \jle p\lor q\lor r = v$. By definition~\eqref{eq:def:sum:order}, we have that (*)~$\varphi_{ps} (a) \le_s \psi_{qs} (b)$ and (**)~$\varphi_{qt} (b) \le_t \psi_{rt} (c)$. Now,
\begin{align*}
\varphi_{uv} (\varphi_{pu} (a)) &= \varphi_{pv} (a)
	= \varphi_{sv} (\varphi_{ps} (a))      &&\text{by the compatibility of $\Phi$}\\
	&\le_v \varphi_{sv} (\psi_{qs} (b))	   &&\text{by the monotonicity of $\varphi_{sv}$ and (*)}\\
	&\le_v \psi_{tv} (\varphi_{qt} (b))    &&\text{by~\ref{S3}}\\
	&\le_v \psi_{tv} (\psi_{rt} (c))       &&\text{by the monotonicity of $\psi_{tv}$ and (**)}\\
	&= \psi_{rv} (c)
	= \psi_{uv} (\psi_{ru} (c)),           &&\text{by the compatibility of $\Psi$.}
\end{align*}
If $u \jl v$, we deduce by~\ref{S2} that $\varphi_{pu} (a) <_u \psi_{ru} (c)$ and if $u = v$, then $\varphi_{uv}= \mathrm{id}_{A_u}=\psi_{uv}$ and therefore $\varphi_{pu} (a) \le_u \psi_{ru} (c)$. In any event, $a \le c$ by definition~\eqref{eq:def:sum:order}. 
\end{proof}

A \emph{partition pair} for a partial ordered set $\m A = \pair{A,\le}$ is a pair $\pair{\otimes, \odot}$ of homotactic left normal bands on $A$ satisfying the following conditions.

\begin{enumerate}[(PS1), leftmargin=*, series=PS]
    \item\label{PS1} If $a\le b$, then $a\otimes c \le b\otimes c$ and $a\odot c \le b\odot c$.
    \item\label{PS2} If $a\odot b \le b\otimes a$, then $a\le b$.
\end{enumerate}

If $\pair{\otimes,\odot}$ is a partition pair for a poset $\m A = \pair{A,\le}$ and $\m I = \pair{I,\lor}$ is the induced join-semilattice on a set of representatives, we denote by $A_p$ the equivalence class of $p$ and $\m A_p = \pair{A_p,\le_p}$ the corresponding subposet of $\m A$. By virtue of Lemma~\ref{lem:LNB:induces:DS}, the family $\Xi = \{\pair{\psi_{pq},\varphi_{pq}}\: A_p\to A_q : {p\jle_\odot q} \text{ in } \m I\}$ given by $\psi_{pq}(a) = a\otimes q$\ and $\varphi_{pq}(a) = a\odot q$ is a well-defined coherent family of parallel maps. And by~\ref{PS1}, all $\psi_{pq}$ and $\varphi_{pq}$ are monotone. That is, $\Xi\:\m I\to \Pos^\parallel$ is a directed system of parallel monotone maps.

\begin{theorem}\label{thm:partition:pair:poset:sum}
Let $\m A = \pair{A,\le}$ be a poset with a partition pair $\pair{\otimes,\odot}$. Then, the directed system of parallel monotone maps $\Xi$ satisfies~\ref{S1} and~\ref{S2} and $\m A$ is its sum. Reciprocally, every directed system of parallel monotone maps satisfying~\ref{S1} and~\ref{S2} is induced by a partition pair on its sum.
\end{theorem}

\begin{proof}
If $\pair{\otimes,\odot}$ is a partition pair for a poset $\m A$ and $\Xi$ is the directed system of parallel monotone maps defined above, then for all $p,q\in I$, $a\in A_p$, $b\in A_q$, and $s = p\lor q$, we have
\[
\varphi_{ps}(a) = a\odot s = a\odot p\odot q = a\odot q = a\odot b
\]
and analogously $\psi_{qs}(b) = b\otimes a$. Hence, in order to show that $\m A$ is the sum of $\Xi$, it would suffice to show that $a\le b$ if and only if $a\odot b\le b\otimes a$. By~\ref{PS2}, we have that if $a\odot b\le b\otimes a$, then $a\le b$. For the reverse implication, notice that, since $\otimes$ and $\odot$ are homotactic and $b \jle_\otimes b\otimes a$, then $b \jle_\odot b\otimes a$, that is $(b\otimes a)\odot b = b\otimes a$. Applying~\ref{PS1} to $a\le b$ we obtain $a = a\otimes a \le b\otimes a$, and again by~\ref{PS1},
\[
a\odot b \le (b\otimes a)\odot b = b\otimes a.
\]
Therefore, by Theorem~\ref{thm:char:sums:of:posets:over:directed:systems}, $\Xi$ satisfies~\ref{S1} and~\ref{S2} and $\m A$ is the sum of $\Xi$.

Suppose now that $\Xi\:\m I\to\Pos^\parallel$ is a directed system satisfying~\ref{S1} and~\ref{S2} and let $\m A$ be its sum. Let us define $a\otimes b = \psi_{ps}(a)$ and $a\odot b = \varphi_{ps}(a)$, where $p,q\in I$ are such that $a\in A_p$ and $b\in A_q$ and $s = p\lor q$. By Lemma~\ref{lem:DS:induces:LNB}, $\otimes$ and $\odot$ are homotactic left normal bands on $A$. The monotonicity of the maps $\psi_{pq}$ and $\varphi_{pq}$ ensures the that $\otimes$ and $\odot$ satisfy~\ref{PS1}. And since the order of $\m A$ is defined by~\eqref{eq:def:sum:order}, then
\[
a \le b \iff \varphi_{ps}(a)\le_s\psi_{qs}(b) \iff a\odot b \le b\otimes a,
\]
that is, $\pair{\otimes,\odot}$ satisfies~\ref{PS2}. In summary, $\pair{\otimes,\odot}$ is a partition pair for $\m A$. The fact that $\Xi$ is its induced directed system can be readily checked.
\end{proof}

Finally, we can combine the above notions and results with the ones of the previous section. Thus, a \emph{metamorphism} $f\: \m A\to \m B$ between two partially ordered algebras of the same type~$\tau$ is an assignment $f\:{\tau\cup\{\le\}}\to \big(B^A\big)^*$ such that the restriction of $f$ to $\tau$ is a metamorphism between the corresponding algebraic reducts and $f^\le = \pair{f^{\le 1},f^{\le 2}}$ is a pair of parallel monotone maps between the corresponding poset reducts. A \emph{directed system of metamorphisms} is a functor $\Xi\:\m I \to \Meta^{\tau,\le}$ where $\m I$ is a join-semilattice, as usual, and $\Meta^{\tau,\le}$ is the category of partially ordered algebras of type~$\tau$ and metamorphisms between them. There are two forgetful functors $\Meta^{\tau,\le}\to\Meta^\tau$ and $\Meta^{\tau,\le}\to\Pos^\parallel$, and the compositions of a directed system of metamorphisms $\Xi$ with them result in a directed system of metamorphisms $\Xi^\alg$ between the corresponding algebraic reducts and a directed system $\Xi^\po$ of parallel monotone maps. The \emph{sum} of $\Xi$ is then the structure $\m A$ whose algebraic and relation reducts are the P\l{}onka sum and relation sum of the corresponding directed systems. We say that $\Xi$ satisfies~\ref{S1} and~\ref{S2} if $\Xi^\po$ does.

A \emph{partition system} of a partially ordered algebra $\m A$ is an assignment $\tau\cup\{\le\}\to O^*$, where $O$ is a set of homotactic left normal bands on $A$, such that the restriction to $\tau$ is a partition system of the algebraic reduct of $\m A$ and ${\le}\mapsto\pair{\otimes,\odot}$, where $\pair{\otimes,\odot}$ is a partition pair for the poset reduct of $\m A$.

The following theorem is an immediate consequence of Theorems~\ref{thm:Plonka:decomposition:meta}, \ref{thm:Plonka:composition:meta}, and~\ref{thm:partition:pair:poset:sum}.

\begin{theorem}\
Every partition system on a partially ordered algebra induces a directed system of metamorphisms of partially ordered algebras satisfying~\ref{S1} and~\ref{S2} such that $\m A$ is its sum. Conversely, every directed system of metamorphisms between partially ordered algebras satisfying~\ref{S1} and~\ref{S2} is induced by a partition system for its sum.
\end{theorem}

\section{Residuated Semigroups Steady over a Semilattice}
\label{sec:fibrant}

In Section~\ref{sec:decompositions} we showed that a balanced residuated semigroup satisfying condition~\ref{cond:H} decomposes into a family $\m A_p$ of (integrally closed) residuated monoids. In Sections~\ref{plonka} and~\ref{sec:sums:of:posets} we developed some general universal algebraic machinery for gluing such families of ordered algebras together. We now apply this machinery to the original motivating case of residuated semigroups, obtaining a general decomposition (Theorem~\ref{thm:steady:over:I:partition:meta}) and composition (Theorem~\ref{thm:construction:meta}) result.

We shall now work at a greater level of generality than in Section~\ref{sec:decompositions}, where we decomposed a balanced residuated semigroup $\m A$ satisfying~\ref{cond:H} into a family of residuated monoids $\m A_p$ indexed by the join semilattice $\ZIDP \m A = \langle \ZIdp \m A, \cdot \rangle$ of central positive idempotents, whose induced order coincides with the restriction of the order of $\m A$ to $\ZIdp \m A$.

Namely, throughout this section we consider a residuated semigroup $\m A$ together with a (non-empty) \emph{sub}semilattice $\m I$ of $\ZIDP \m A$ such that for each $a \in A$ the following element exists:
\[
  u_a := \max\{p\in I : pa = a\} = \max\{p\in I : ap = a\}.
\]
We may equivalently define $u_a$ as $u_a = \max \{p\in I : p\le a\rd a\}$ or $u_a = \max \{p\in I : p\le a\ld a\}$. We call such a residuated semigroup \emph{balanced over $\m I$}. We shall use the notation
\[
  A_p := \{ a \in A : u_a = p \}.
\]
The following proposition shows that this is consistent with our previous usage of this notation.

\begin{proposition}
A residuated semigroup $\m A$ is balanced over $\IDP \m A$ if and only if it is balanced. In that case, $u_p = 1_p$ and $A_p = \{ a \in A : a \ld a = p \} = \{ a \in A : a \rd a = p \}$.
\end{proposition}

\begin{proof}
  If $\m A$ is balanced, then $\ZIDP \m A = \IDP \m A$ and for $\m I := \IDP \m A$ we have $u_p = 1_p = a \ld a = a \rd a$. Conversely, if $\m A$ is balanced over $\m I := \IDP \m A$, then $\IDP \m A = \ZIDP \m A$, since being balanced over $\m I$ implies that $\m I \le \ZIDP \m A$. Moreover, $u_a$ is the largest positive idempotent such that $p \le a \ld a$ and also the largest positive idempotent such that $p \le a \rd a$, so $a \ld a = u_a = a \rd a$.
\end{proof}

\begin{remark}\label{rem:on:us}
If $\m A$ is balanced over $\m I$, then all self-residuals of $\m A$ are positive: $u_a\cdot a = a = a\cdot u_a$ implies that $u_a \le a\rd a$ and $u_a \le a\ld a$, so $a \ld a$ and $a \rd a$ are positive because $u_a$ is. Nonetheless, $\m A$~need not be balanced: as an extreme case, take $I = \{ 1 \}$ in any residuated monoid with a global identity $1$.
\end{remark}

  In order to decompose a residuated semigroup $\m A$ balanced over $\m I$, we wish to define a left normal band $\odot\: A^2\to A$ on $A$ such that its associated partition is $\{A_p : p \in I\}$. But recall property~\ref{*}, which states that $a \odot b = a \odot c$ if $b$ and $c$ belong to the same class in the partition associated with $\odot$. If this partition is $\{A_p : p \in I\}$, then $a \odot b = a \odot u_b$. There are two natural choices for such a function $\odot$:
\[
a\odot b := a\cdot u_b \qquad\text{and}\qquad a\otimes b := a\rd u_b = u_b\ld a.
\]
In general, these operations are not left normal bands because they may fail to be associative. However, they are left normal bands if the following equations are satisfied:
\[
u_{x\cdot u_y} \approx u_x\cdot u_y \qquad\text{and}\qquad u_{x\rd u_y} \approx u_x\cdot u_y.
\]
We may also rewrite these conditions more suggestively as follows. Consider the map $\pi\: {A\to I}$ defined by $\pi(a) := u_a$. Consider also, for $p \in I$, the closure operator $\gamma_p$ and the interior operator $\delta_p$ on $\pair{A,\le}$ given by $\gamma_p(a) := ap$ and $\delta_p(a) := a\rd p$ (see Lemma~\ref{lem:Ap:closure:meets:joins}). The above equations are equivalent to the following equalities being satisfied for all $a\in A$ and $p\in I$:
\[
\pi(\gamma_p(a)) = \gamma_p(\pi(a))\qquad\text{and}\qquad \pi(\delta_p(a)) = \gamma_p(\pi(a)).
\]

\begin{proposition}\label{prop:on:us}
  For all $a, b \in A$ we have $u_a \ld a = a = a \rd u_a$ and
\[
u_a \le u_{ab},\quad u_b \le u_{ab},\quad u_a\le u_{a\rd b},\quad
u_b\le u_{a\rd b},\quad u_a\le u_{b\ld a},\quad u_b\le u_{b\ld a}.
\]    
  Moreover, $u_p = p$ for each $p \in I$. If $\m A$ contains a global identity $1$, then $1$ is the least element of $\m I$.
\end{proposition}

\begin{proof}
The inequality $a\rd u_a \le a$ holds because $u_a$ is positive, and $a\cdot u_a = a$ implies that $a\le a\rd u_a$. Therefore, $a\rd u_a = a$. Since $u_a$ is central, $u_a\ld a = a$ too. We will prove only two of the displayed inequalities, as the other proofs are similar. Notice that $u_b\cdot ab = a(u_b\cdot b) = ab$, and therefore $u_b \le u_{ab}$, by the maximality of $u_{ab}$. Analogously, $u_a\cdot (a\rd b)b\le u_a \cdot a = a$ and therefore $u_a\cdot (a\rd b) \le a\rd b$, whence we obtain that $u_a\le u_{a\rd b}$ by the maximality of $u_{a\rd b}$.

If $p\in I$, then $u_p = p$, since $p\cdot p = p$ and for all $q\in I$, if $qp = p$, then $q \le p\rd p = p$. As a consequence, for every $a\in A$ we have that $u_{u_a} = u_a$. If $\m A$ contains a global identity $1$, then $u_1 = u_1\cdot 1 = 1$, and $1 \in I$. Since $1$ is the least element of $\ZIdp \m A$, it is also least element of $\m I$.
\end{proof}

\begin{proposition}\label{prop:pi:gamma:delta:over:I:left:normal:bands}
Let $\m A$ be a residuated semigroup balanced over $\m I$. Then the equations
\begin{equation}\label{eq:pi:gamma:pi:delta}
u_{a\cdot u_b} = u_a\cdot u_b \qquad\text{and}\qquad u_{a\rd u_b} = u_a\cdot u_b    
\end{equation}
hold for all $a,b\in A$ if and only if the maps $\odot,\otimes\: {A^2\to A}$ defined by $a\odot b := a\cdot u_b$ and $a\otimes b := a\rd u_b$ are homotactic left normal bands on $A$. In that case, $a\jle_\odot b$ if and only if $u_a\le u_b$ for $a, b \in A$, so $\jle_\odot$ and $\le$ restricted to $I$ coincide. Furthermore, $\pair{\otimes,\odot}$ is a partition pair of the poset reduct of $\m A$.
\end{proposition}

\begin{proof}
Suppose that the equalities~\eqref{eq:pi:gamma:pi:delta} hold in $\m A$. Then, the following computations show that both~$\odot$ and~$\otimes$ are left normal bands:
\begin{gather*}
a\odot a = a\cdot u_a = u_a\cdot a = a,\\
a\odot (b\odot c) = a\cdot u_{b\cdot u_c} = a\cdot u_b\cdot u_c = (a\odot b)\odot c,\\
a\odot b\odot c = a\cdot u_b\cdot u_c = a\cdot u_c\cdot u_b = a\odot c\odot b,\\
a\otimes a = a\rd u_a = a,\\
a\otimes(b\otimes c) = a\rd u_{b\rd u_c} = a\rd (u_b\cdot u_c)  = a\rd (u_c\cdot u_b) = (a\rd u_b)\rd u_c = (a\otimes b)\otimes c,\\
a\otimes b\otimes c = (a\rd u_b)\rd u_c = a\rd (u_c\cdot u_b) = a\rd (u_b\cdot u_c) = (a\rd u_c)\rd u_b = a\otimes c\otimes b.
\end{gather*}

Reciprocally, if~$\odot$ and~$\otimes$ are left normal bands on $A$, in particular they are associative. Using the inequalities of Proposition~\ref{prop:on:us} we have $u_b=u_{u_b}\le u_{a\cdot u_b}$, hence
\[
u_a\cdot u_b \le u_a\cdot u_{a\cdot u_b} = u_a \odot (a\odot b) =(u_a\odot a)\odot b = (u_a\cdot u_a) \cdot u_b = u_a\cdot u_b.
\]
As for the other equality,
\begin{align*}
(u_a\cdot u_b)\rd u_{a\rd u_b} &= (u_a\cdot u_b)\otimes (a\rd u_b) = (u_a\cdot u_b)\otimes (a\otimes b) 
= ((u_a\cdot u_b)\otimes a) \otimes b \\
&= ((u_a\cdot u_b)\rd u_a) \rd u_b = (u_a\cdot u_b)\rd (u_b \cdot u_a) = (u_a\cdot u_b)\rd (u_a \cdot u_b)
= u_a\cdot u_b,
\end{align*}
whence we deduce that $u_{a\rd u_b}\le u_a\cdot u_b\cdot u_{a\rd u_b} \le u_a\cdot u_b \le u_{a\rd u_b}$, by residuation and $u_a\cdot u_b = u_a\cdot u_{u_b} \le u_{a\rd u_b}\cdot u_{a\rd u_b} = u_{a\rd u_b}$, by the inequalities of Proposition~\ref{prop:on:us}. Now, given $a,b\in A$, we have
\begin{gather*}
a\jle_\odot b \iff b\odot a = b \iff b\cdot u_a = b \iff u_a \le u_b,\\
a\jle_\otimes b \iff b\otimes a = b \iff b\rd u_a = b \iff b\le b\rd u_a \iff b\cdot u_a\le b \iff u_a\le u_b.
\end{gather*}
by the maximality of $u_b$. In particular, for all $p,q\in I$ we have that
\[
p\jle_\odot q \quad\iff\quad p=u_p\le u_q = q \quad\iff\quad p\jle_\otimes q.
\]

Finally, both $\odot$ and $\otimes$ satisfy~\ref{PS1} because the multiplication is monotone in both arguments and the residuals are monotone in the numerator. Also~\ref{PS2} is satisfied since for all $a,b\in A$, if $a\odot b\le b\otimes a$, then
\[
a\le a\cdot u_b = a\odot b \le b\otimes a = b\rd u_a \le b.\tag*{\QED}
\]\let\QED\relax
\end{proof}

If, moreover, we want these left normal bands to be compatible with the operations of the residuated semigroup, more is required. We say that $\m A$ is \emph{steady over} $\m I$ if it is balanced over $\m I$ and the three following equalities hold for all $a,b\in A$:
\begin{enumerate}[({St}1),leftmargin=*,series=St]
    \item\label{St1} $u_{ab} = u_a\cdot u_b$,
    \item\label{St2} $u_{a\rd b} = u_a\cdot u_b$,
    \item\label{St3} $u_{b\ld a} = u_a\cdot u_b$.
\end{enumerate}
We say that $\m A$ is \emph{steady} if it is steady over $\IDP\m A$, or equivalently if $\Idp \m A = \ZIdp \m A$ and $\m A$ is steady over $\ZIDP \m A$. Because being balanced is equivalent to being balanced over $\IDP \m A$, being steady implies being balanced. Because steadiness is an equational condition, the class of steady residuated semigroups is a sub-po-variety of the class of (balanced) residuated semigroups. The steadiness conditions are in fact slightly redundant, as we prove in the next lemma.

\begin{lemma}\label{lem:St2:or:St3:imply:St1}
Let $\m A$ be a residuated semigroup balanced over $\m I$. If any of the equalities~\ref{St2} or~\ref{St3} hold for all $a,b\in A$, then \ref{St1} also holds for all $a,b\in A$.
\end{lemma}

\begin{proof}
Suppose that $\m A$ satisfies~\ref{St2} and let $a,b\in A$ be arbitrary elements. From Proposition~\ref{prop:on:us} we obtain $u_b\le u_{ab}$, and by the monotonicity of the product and several applications of~\ref{St2} we have
\[
u_a\cdot u_b \le u_a\cdot u_{ab} = u_{a\rd (ab)} = u_{(a\rd b)\rd a} = u_{a\rd b}\cdot u_a = u_a\cdot u_b\cdot u_a = u_a\cdot u_a\cdot u_b
=  u_a\cdot u_b.
\]
Therefore $u_{ab}\le u_a\cdot u_{ab}=u_a\cdot u_b$. The reverse inequality follows from maximality of $u_{ab}$, hence
$\m A$ satisfies~\ref{St1}. The proof that~\ref{St3} implies~\ref{St1} is entirely analogous.
\end{proof}

\begin{remark}\label{rem:H:does:not:imply:St}
Every balanced residuated semigroup~$\m A$ that satisfies~\ref{St1} with respect to $\IDP\m A$ also satisfies~\ref{H1}. Indeed, if $a,b\in A$ are such that $1_a = 1_b$, then $1_{ab} = 1_a\cdot 1_b = 1_a\cdot 1_a = 1_a$, by a direct application of~\ref{St1}. Analogously, both~\ref{St2} and~\ref{St3} imply the equivalent conditions~\ref{H2} and~\ref{H3}. On the other hand, condition~\ref{cond:H} does not imply steadiness. Indeed, the residuated poset of Example~\ref{ex:RP:H123} is not steady: it fails~\ref{St1} for $\IDP\m A$, and therefore also~\ref{St2} and~\ref{St3}, because $pa = b$ and $1_{pa} = 1_b = q \neq p\cdot 1 = 1_p\cdot 1_a$.
\end{remark}

\begin{example}\label{exam:St1:weaker:St2:and:St3}
The four-element linearly-ordered commutative and idempotent residuated semigroup $\m A$ with global identity~$1$ such that $\bot \le 1 \le p \le q$ satisfies that $\ZIDP\m A = \IDP\m A = \pair{\{1,p,q\},\cdot}$. We could take $\m I = \IDP\m A$ and therefore $u_a = 1_a$ for every $a\in A$. One can readily check that it satisfies~\ref{St1}, but it fails~\ref{St2} because
\[
1_{1\rd p} = 1_\bot = q \neq p = 1\cdot p = 1_1\cdot 1_p,
\]
as well as~\ref{St3}, because $\m A$ is commutative. This shows that~\ref{St1} is strictly weaker than~\ref{St2} and~\ref{St3}.
\end{example}

We will prove that for every residuated semigroup $\m A$ that satisfies~\ref{St1}, the left normal band $\odot$ satisfies~\hyperref[PF4s]{(PF4$^\cdot$)} and~\hyperref[PF5s]{(PF5$^\cdot$)}. That is, $\odot$ is a partition function for the semigroup reduct of $\m A$. The residuals of residuated semigroups resist such a direct approach. If $\otimes$ were a partition function with respect to the residuals, then $\m A$ would satisfy $(u_a\rd u_a)\rd u_b = (u_a\rd u_b)\rd (u_a\rd u_b)$, for all $a,b\in A$, as a particular case of~\hyperref[PF4s]{(PF4$^\rd$)}. However, as the next lemma shows, in that case $\m I$ would be trivial. In particular, all balanced residuated semigroups that satisfy this for $\m I = \IDP\m A$ are integrally closed. Therefore, this condition is overly restrictive and will need to be relaxed.

\begin{lemma}\label{lem:Plonka:sums:are:not:enough}
Let $\m A$ be a residuated semigroup and $\m I\le\ZIDP\m A$. Then, $\m A$ satisfies $(u_a\rd u_a)\rd u_b = (u_a\rd u_b)\rd (u_a\rd u_b)$ for all $a,b\in A$ if and only if $\m I$ is trivial.   
\end{lemma}

\begin{proof}
First, recall that for every $p\in I$, $u_p = p$. Suppose that $(u_a\rd u_a)\rd u_b = (u_a\rd u_b)\rd (u_a\rd u_b)$ for all $a,b\in A$. Then, for all $p,q\in I$,
\[
pq = u_p\cdot u_q \le u_{p\rd q}\cdot u_{p\rd q} = u_{p\rd q} \le (p\rd q)\rd (p\rd q) = (p\rd p)\rd q = p\rd q,
\]
where the assumption is used in the second-to-last equality.
Hence, $p \le pq = pqq \le p$ and therefore $pq = p$. By a symmetric argument $pq = q$. Hence, $I$ contains only one element. The reverse implication is immediate.
\end{proof}

Given a semilattice $\m I \le \ZIDP \m A$, we say that a residuated semigroup $\m A$ is \emph{fibrant over $\m I$} if the operations $\otimes,\odot\: {A^2\to A}$ defined in the previous proposition are indeed homotactic left residuated bands on $A$ and the assignment $\cdot\mapsto\pair{\odot,\odot,\odot}$, $\ld\mapsto\pair{\otimes,\odot,\otimes}$, $\rd\mapsto\pair{\otimes,\otimes,\odot}$, and ${\le}\mapsto \pair{\otimes,\odot}$ defines a partition system for $\m A$. We use the same terminology for residuated monoids, adjoining $1\mapsto \pair{\odot}$ to the partition system.

\begin{theorem}\label{thm:steady:over:I:partition:meta}
Let $\m A$ be a residuated semigroup steady over $\m I\le \ZIDP\m A$. Then $\m A$ is fibrant over $\m I$. Moreover, $\m A$ is the sum of the directed system of metamorphisms $\Xi = \{\xi_{pq}\: {\m A_p\meta\m A_q} : {p\le q} \text{ in } \m I\}$ between the residuated monoids $\m A_p$, given by 
\begin{gather*}
\xi^1_{pq} = \pair{\varphi_{pq}},\qquad
\xi^\cdot_{pq} = \pair{\varphi_{pq},\varphi_{pq},\varphi_{pq}},\\
\xi^\ld_{pq} = \pair{\psi_{pq}, \varphi_{pq}, \psi_{pq}},\qquad
\xi^\rd_{pq} = \pair{\psi_{pq}, \psi_{pq}, \varphi_{pq}},\qquad
\xi^\le_{pq} = \pair{\psi_{pq},\varphi_{pq}},
\end{gather*}
where the maps $\varphi_{pq},\psi_{pq}\: A_p\to A_q$ for $p \jle q$ in $\m I$ are defined by
\[
  \varphi_{pq}(a) := aq, \qquad \psi_{pq}(a) := a\rd q.
\]
  If $\m I = \ZIDP \m A$, then each $\m A_p$ is integrally closed. If $\m I = \ZIDP \m A$ and $\m A$ is moreover square-decreasing (i.e., satisfies the inequation $x \cdot x \le x$), then each $\m A_p$ is integral.
\end{theorem}

\begin{proof}
If $\m A$ is steady over $\m I$, then in particular for all $a,b\in A$, we have that $u_{a\cdot u_b} = u_a\cdot u_{u_b} = u_a\cdot u_b$ and $u_{a\rd u_b} = u_a\cdot u_{u_b} = u_a\cdot u_b$. Therefore, by virtue of Proposition~\ref{prop:pi:gamma:delta:over:I:left:normal:bands}, $\otimes$ and $\odot$ are homotactic left normal bands on~$A$ whose induced partition is $\{ A_p : p\in I\}$. Furthermore, $\pair{\otimes,\odot}$ is a partition pair for the poset reduct of $\m A$. Let us show that the remaining properties of a partition system are also satisfied. As for~\hyperref[PF4s]{(PF4$^\cdot$)}, indeed
\[
(a_1\cdot a_2)\odot b = (a_1\cdot a_2)\cdot u_b = a_1\cdot a_2\cdot u_b\cdot u_b = a_1 \cdot u_b\cdot a_2\cdot u_b = (a_1\odot b)\cdot (a_2\odot b).    
\]
Concerning~\hyperref[PF5s]{(PF5$^\cdot$)}, we have that $b\odot (a_1\cdot a_2) = b\cdot u_{a_1\cdot a_2}= b\cdot u_{a_1}\cdot u_{a_2} = b\odot a_1\odot a_2$, by~\ref{St1}. The property~\hyperref[PF4s]{(PF4$^\rd$)} follows from the definitions of~$\otimes$ and~$\odot$:
\begin{align*}
\begin{split}
(a_1\rd a_2)\otimes b &= (a_1\rd a_2)\rd u_b = a_1\rd (u_b\cdot a_2) = a_1\rd (a_2\cdot u_b) = a_1\rd (a_2\cdot u_b\cdot u_b) \\
&= (a_u\rd u_b)\rd (a_2\cdot u_b) = (a_1\otimes b)\rd (a_2\odot b).    
\end{split}
\end{align*}
Analogously for~\hyperref[PF4s]{(PF4$^\ld$)}. Let us finally prove that~\hyperref[PF5s]{(PF5$^\rd$)} and~\hyperref[PF5s]{(PF5$^\ld$)} are also satisfied. For the first equality
\[
b\otimes (a_1\rd a_2) = b\rd u_{a_1\rd a_2} = b\rd (u_{a_1}\cdot u_{a_2}) = b\rd (u_{a_2}\cdot u_{a_1}) = (b\rd u_{a_1})\rd u_{a_2} = b\otimes a_1\otimes a_2,    
\]
which is a consequence of~\ref{St2}. The second equality, $b\otimes (a_1\ld a_2) = b\otimes a_1\otimes a_2$, is a consequence of~\ref{St3}. If $\m A$ is a residuated monoid, then we can also check that it satisfies~\hyperref[PF5w]{(PF5$^1)$}, since $b\odot 1 = b\cdot u_1 = b\cdot 1 = b$. Therefore $\m A$ is the sum of the given directed system of metamorphisms by Theorem~\ref{thm:Plonka:decomposition:meta} and Theorem~\ref{thm:partition:pair:poset:sum}. The fibers $\m A_p$ are residuated semigroups because they are subalgebras of $\m A$, and $p$ is a global identity on $\m A_p$ by the definition of $\m A_p$.

Now suppose that $\m I = \ZIDP \m A$. Because $\m A$ is steady over $\m I$, it is balanced over $\m I$, or in other words, balanced. For each $a \in A_p$ with $p \in I$ we therefore have $a \ld_p a = a \ld^{\m A} a \in \ZIdp \m A = I$. But $I \cap A_p = \{ p \}$, so $a \ld_p a = p$ and $\m A_p$ is integrally closed.

Suppose moreover that $\m A$ is square-decreasing and consider $a \in A_p$ with $p \in I = \ZIdp \m A$. To prove that $\m A_p$ is integral, we need to show that $a \le_p p$, i.e., $a \le^{\m A} p$. By the previous paragraph, $p = a \ld^{\m A} a$, so this inequality is equivalent to $a \le^{\m A} a \ld^{\m A} a$, which follows by residuation from the square-decreasing property.
\end{proof}

\begin{remark}
For the reader's convenience, let us record explicitly how the operations of $\m A$ and the partial order are computed in the sum of the semilattice directed system of metamorphisms $\Xi$ in the above theorem. Given $a \in A_p$ and $b \in A_q$ and taking $r := p \vee q$ in $\m I$,
\[
  a \cdot^{\m A} b = \varphi_{pr}(a) \cdot_r \varphi_{qr}(b), \qquad a \ld^{\m A} b = \varphi_{pr}(a) \ld_{r} \psi_{qr}(b), \qquad a \rd^{\m A} b = \psi_{pr}(a) \rd_{r} \varphi_{qr}(b)
\]
  and
\[
  a \le b \iff \varphi_{pr}(a) \le_{r} \psi_{qr}(b).
\]
\end{remark}

\begin{example} \label{ex:iterated:decomposition}
As we noticed in Remark~\ref{rem:H:does:not:imply:St}, the commutative and idempotent residuated semigroup of Example~\ref{ex:RP:H123} is not steady. But it is steady over $\m I = \pair{\{1,p\},\cdot}$, where $u_a = u_1 = 1$ and $u_q = u_p = u_b = u_\bot = p$, and therefore it is fibrant over $\m I$. In turn, the second fiber is steady, and therefore fibrant over its idempotents. All these fibers and the maps between them are depicted in Figure~\ref{fig:iterated:decomposition:1}.
\end{example}

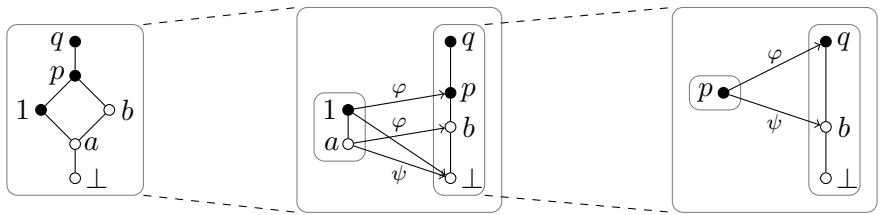
\begin{figure}[ht]\centering
\begin{tikzpicture}[baseline=0pt]
\draw [color=gray, rounded corners](-4,-0.5) rectangle (0,4.5);
\node(4) at (-2,4)[i, label=left:$q$]{};
\node(3) at (-2,3)[i, label=left:$p$]{} edge (4);
\node(2) at (-1,2)[label=right:$b$]{} edge (3);
\node(1) at (-3,2)[i, label=left:$1$]{} edge (3);
\node(0) at (-2,1)[label=right:$a$]{} edge (1) edge (2);
\node(-1) at (-2,0)[label=right:$\bot$]{} edge (0);

\draw [color=gray, rounded corners](4.5,-1) rectangle (10.5,5);
\draw [color=gray, rounded corners](5,.5) rectangle (6.5,2.5);
\draw [color=gray, rounded corners](8.5,-0.5) rectangle (10,4.5);
\node(4) at (9,4)[i, label=right:$q$]{};
\node(3) at (9,2.5)[i, label=right:$p$]{} edge (4);
\node(2) at (9,1.5)[label=right:$b$]{} edge (3);
\node(1) at (6,2)[i, label=left:$1$]{};
\node(0) at (6,1)[label=left:$a$]{} edge (1);
\node(-1) at (9,0)[label=right:$\bot$]{} edge (2);

\draw [->] (1) -- node[draw=none, minimum size=0pt, label=above:{\scriptsize $\varphi$}]{} (3);
\draw [->,dashed] (1) -- node[draw=none, minimum size=0pt]{} (-1);
\draw [->] (0) -- node[draw=none, minimum size=0pt, label=above:{\scriptsize $\varphi$}]{} (2);
\draw [->,dashed] (0) -- node[draw=none, minimum size=0pt, label=below:{\scriptsize $\psi$}]{} (-1);

\draw [color=gray, rounded corners](15.5,-1) rectangle (21.5,5);
\draw [color=gray, rounded corners](16,2) rectangle (17.5,3);
\draw [color=gray, rounded corners](19.5,-0.5) rectangle (21,4.5);
\node(4) at (20,4)[i, label=right:$q$]{};
\node(3) at (17,2.5)[i, label=left:$p$]{};
\node(2) at (20,1.5)[label=right:$b$]{} edge (4);
\node(-1) at (20,0)[label=right:$\bot$]{} edge (2);

\draw [->] (3) -- node[draw=none, minimum size=0pt, label=above:{\scriptsize $\varphi$}]{} (4);
\draw [->,dashed] (3) -- node[draw=none, minimum size=0pt, label=below:{\scriptsize $\psi$}]{} (2);

\draw [dashed] (0,4.5) -- (4.5,5);
\draw [dashed] (0,-0.5) -- (4.5,-1);
\draw [dashed] (10,4.5) -- (15.5,5);
\draw [dashed] (10,-0.5) -- (15.5,-1);
\end{tikzpicture}

\caption{Iterated decomposition of the residuated semigroup of Example~\ref{ex:RP:H123}.}
\label{fig:iterated:decomposition:1}
\end{figure}

\begin{example}
The residuated semigroup of Example~\ref{exam:St1:weaker:St2:and:St3} is not steady, as we saw, but it is steady over $\m I = \pair{\{1,p\},\cdot}$. In this case $u_1 = 1$ and $u_a = p$, for every other $a\in A$. Therefore, $\m A$ is fibrant over~$\m I$. In turn, the second fiber is steady, and therefore it is also fibrant over its idempotents. All the fibers of this situation and the corresponding maps are depicted in Figure~\ref{fig:iterated:decomposition:2}.
\end{example}

\begin{figure}[ht]\centering
\begin{tikzpicture}[baseline=0pt]
\draw [color=gray, rounded corners](-0.5,-0.5) rectangle (1,3.5);
\node(4) at (0,3)[i, label=right:$q$]{};
\node(3) at (0,2)[i, label=right:$p$]{} edge (4);
\node(2) at (0,1)[i, label=right:$1$]{} edge (3);
\node(1) at (0,0)[label=right:$\bot$]{} edge (2);

\draw [color=gray, rounded corners](5.5,-1) rectangle (11.5,4);
\draw [color=gray, rounded corners](6,.5) rectangle (7.5,1.5);
\draw [color=gray, rounded corners](9.5,-0.5) rectangle (11,3.5);
\node(4) at (10,3)[i, label=right:$q$]{};
\node(3) at (10,2)[i, label=right:$p$]{} edge (4);
\node(2) at (7,1)[i, label=left:$1$]{};
\node(1) at (10,0)[label=right:$\bot$]{} edge (3);

\draw [->] (2) -- node[draw=none, minimum size=0pt, label=above:{\scriptsize $\varphi$}]{} (3);
\draw [->,dashed] (2) -- node[draw=none, minimum size=0pt, label=below:{\scriptsize $\psi$}]{} (1);

\draw [color=gray, rounded corners](15.5,-1) rectangle (21.5,4);
\draw [color=gray, rounded corners](16,1) rectangle (17.5,2);
\draw [color=gray, rounded corners](19.5,-0.5) rectangle (21,3.5);
\node(4) at (20,3)[i, label=right:$q$]{};
\node(3) at (17,1.5)[i, label=left:$p$]{};
\node(1) at (20,0)[label=right:$\bot$]{} edge (4);

\draw [->] (3) -- node[draw=none, minimum size=0pt, label=above:{\scriptsize $\varphi$}]{} (4);
\draw [->,dashed] (3) -- node[draw=none, minimum size=0pt, label=below:{\scriptsize $\psi$}]{} (1);

\draw [dashed] (1,3.5) -- (5.5,4);
\draw [dashed] (1,-0.5) -- (5.5,-1);
\draw [dashed] (11,3.5) -- (15.5,4);
\draw [dashed] (11,-0.5) -- (15.5,-1);
\end{tikzpicture}

\caption{Iterated decomposition of the residuated semigroup of Example~\ref{exam:St1:weaker:St2:and:St3}.}
\label{fig:iterated:decomposition:2}
\end{figure}
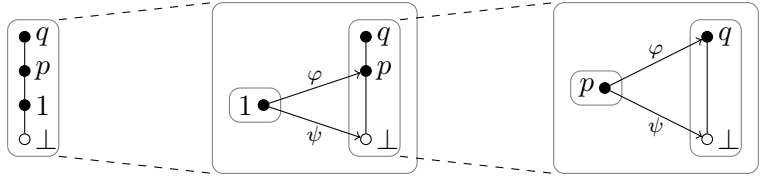

The next proposition shows that steadiness over $\m I$ is not only a sufficient but also necessary for a residuated semigroup balanced over $\m I$ to be fibrant over $\m I$.

\begin{proposition}\label{prop:typically:fibrant:implies:steady}
If a residuated semigroup is balanced over $\m I \le \ZIDP \m A$ and fibrant over $\m I$, then it is steady over~$\m I$ too.
\end{proposition}

\begin{proof}
By virtue of Lemma~\ref{lem:St2:or:St3:imply:St1}, all we need to check is that the equations~\ref{St2} and~\ref{St3} are satisfied under these hypotheses. Now, using~\hyperref[PF5s]{PF5$^\rd$}, we have
\[
b\rd (u_{a_1}\cdot u_{a_2}) = b\rd (u_{a_2}\cdot u_{a_1}) = (b\rd u_{a_1})\rd u_{a_2}
= (b\otimes a_1) \otimes a_2 = b\otimes (a_1\rd a_2) = b\rd u_{a_1\rd a_2},
\]
which for the particular value $b = u_{a_1}\cdot u_{a_2}$ would be
\[
u_{u_{a_1}\cdot u_{a_2}} \le (u_{a_1}\cdot u_{a_2})\rd (u_{a_1}\cdot u_{a_2}) = 
(u_{a_1}\cdot u_{a_2})\rd u_{a_1\rd a_2}.    
\]
Hence, by the monotonicity of the product and applying residuation, we deduce that
\[
u_{a_1}\cdot u_{a_2} \le u_{a_1\rd a_2}\cdot u_{a_1\rd a_2} = u_{a_1\rd a_2} \le u_{u_{a_1}\cdot u_{a_2}} \cdot u_{a_1\rd a_2} \le u_{a_1}\cdot u_{a_2}.
\]
That is, $u_{a_1\rd a_2} = u_{a_1}\cdot u_{a_2}$. Analogously, $u_{a_2\ld a_1} = u_{a_1}\cdot u_{a_2}$.
\end{proof}

The next result is the constructive part of the structural theorem for steady residuated semigroups.

\begin{theorem}\label{thm:construction:meta}
Let $\Xi=\{\xi_{pq}\:\m A_p\meta\m A_q : p\jle q \text{ in }\m I\}$ be a directed system of metamorphisms of residuated monoids of the form
\begin{gather*}
\xi^1_{pq} = \pair{\varphi_{pq}},\qquad
\xi^\cdot_{pq} = \pair{\varphi_{pq},\varphi_{pq},\varphi_{pq}},\\
\xi^\ld_{pq} = \pair{\psi_{pq}, \varphi_{pq}, \psi_{pq}},\qquad
\xi^\rd_{pq} = \pair{\psi_{pq}, \psi_{pq}, \varphi_{pq}},\qquad
\xi^\le_{pq} = \pair{\psi_{pq},\varphi_{pq}},
\end{gather*}
satisfying~\ref{S1} and~\ref{S2}. Then the sum $\m A$ of $\Xi$ is a residuated semigroup, which is a monoid if $\m I$ has a least element. We may identify $\m I$ with $\pair{\{1_p : p\in I\},\cdot^\m A}\le\ZIDP\m A$. Then $\m A$ is steady (hence fibrant) over $\m I$. If each $\m A_p$ is balanced, then so is $\m A$. If each $\m A_p$ is integrally closed, then $\m I = \ZIDP \m A$. If each $\m A_p$ is integral, then $\m A$ is square-decreasing (i.e., satisfies the inequation $x \cdot x \le x$).
\end{theorem}

\begin{proof}
Let $\m A$ be the sum of $\Xi$. The conditions on the components $\xi^1$ and $\xi^\cdot$ imply that $\varphi_{pq}$ is a homomorphism between the monoid reducts of $\m A_p$ and $\m A_q$ for all $p\jle q$ in $\m I$. Therefore, $\pair{A,\cdot^\m A}$ is a semigroup and, if $\m I$ contains a least element $\bot$, then $1^\m A = 1^{\m A_\bot}$ and $\pair{A, \cdot^\m A, 1^\m A}$ is a monoid. Moreover, $\pair{A,\le^\m A}$ is a poset by Proposition~\ref{thm:partition:pair:poset:sum}. For residuation (checking only~$\ld\!^\m A$\,), let $a\in A_p$, $b\in A_q$, $c\in A_r$, $s=p\vee q$, $t=p\vee r$, and $u=s\vee r = q\vee t$. Then, $a\ld\!^\m A c \in A_t$ and
\[
\psi_{tu}(a\ld\!^\m A c) = \psi_{tu}\big(\varphi_{pt}(a)\ld_t\psi_{rt}(c)\big)
= \varphi_{tu}\varphi_{pt}(a)\ld_u\psi_{tu}\psi_{rt}(c)
= \varphi_{pu}(a)\ld_u\psi_{ru}(c).
\]
Hence,
\begin{align*}
    a\cdot b\le c&\iff \varphi_{ps}(a)\cdot_s\varphi_{qs}(b)\le c
    \iff \varphi_{su}(\varphi_{ps}(a)\cdot_s\varphi_{qs}(b))\le \psi_{ru}(c) \\
    &\iff \varphi_{pu}(a)\cdot_u\varphi_{qu}(b)\le_u \psi_{ru}(c)
    \iff \varphi_{qu}(b)\le_u  \varphi_{pu}(a)\ld_u\psi_{ru}(c) \\
    &\iff \varphi_{qu}(b)\le_u \psi_{tu}(a\ld\!^\m A c) 
    \iff b\le a\ld c.
\end{align*}

We now show that $\m I \le \ZIDP \m A$, i.e., that all the identities of the fibers are positive and central in $\m A$. Indeed, if $p,q\in I$, $r = p\lor q$, and $a\in A_p$, then $a\le^\m A \varphi_{pr}(a) = \varphi_{pr}(a)\cdot_r 1_r = \varphi_{pr}(a)\cdot_r\varphi_{qr}(1_q) = a\cdot^\m A 1_q$, and analogously $a\le^\m A 1_q\cdot^\m A a$. Moreover, $a \cdot^{\m A} 1_q = \varphi_{pr}(a) \cdot_r \varphi_{pr}(1_q) = \varphi_{pr}(a) \cdot_r 1_r = \varphi_{pr}(a) = 1_r \cdot_r \varphi_{pr}(a) = \varphi_{pr}(1_q) \cdot_r \varphi_{pr}(a) = 1_q \cdot^{\m A} a$.

We show that $\m A$ is balanced over $\m I$. That is, for each element $a \in A_p$ we show that there is a largest element $u_a \in I$ such that $u_a \cdot^{\m A} a \le^{\m A} a$. Because the elements of $I$ have the form $1_q$, it suffices to show that the largest $q \in I$ such that $1_q \cdot^{\m A} a \le^{\m A} a$ is $q := p$. Clearly $1_p \cdot^{\m A} a = a$. Conversely, if $1_q \cdot^{\m A} a \le^{\m A} a$, then $\varphi_{pr}(a) \le^{\m A} a$ for $r := p \vee q$ in $\m I$. By the definition of $\le^{\m A}$, this inequality is equivalent to $\varphi_{pr}(a) \le_r \psi_{pr}(a)$, which by \ref{S2} holds only if it is not the case that $p \jl r$. Because $r = p \vee q$, this condition is equivalent to $q \jle p$ in $\m I$.

To show that $\m A$ is steady over $\m I$, consider $a \in A_p$ and $b \in A_q$ and take $r := p \vee q$ in $\m I$. Then by the previous paragraph and the definition of operations in $\m A$, we have $u_{ab} = u_{a \ld b} = u_{a \rd b} = r$ and also $u_a \cdot u_b = r$, so the conditions \ref{St1}--\ref{St3} are satisfied.

Suppose now that each $\m A_p$ is balanced and consider some $a\in A_p$ with $p\in I$. Then $a\ld\!^\m A a = a\ld_p\,a = a\rd\!_p\, a = a\rd^\m A a$ and $1_p\le_p a \ld a = a\rd\!_p\, a$, so $1_p \le^\m A a \ld^{\m A} a =  a\rd\!^\m A a$. Because $1_p$ is positive in $\m A$, so is $a \ld\!^\m A a = a \rd^\m A a$, hence all self-residuals of $\m A$ are positive and $\m A$ is balanced.

Suppose that all the fibers are integrally closed and consider $a \in \ZIdp \m A$. Then $a \in A_p$ for some $p \in I$. Because $\m A_p$ is a subalgebra of $\m A$, the element $a$ is a central positive idempotent of $\m A_p$. But $\m A_p$ is integrally closed and therefore it only contains one central positive idempotent element, namely $1_p$. Thus $a = 1_p \in I$. Finally, if each $\m A_p$ is integral, then for each $a \in A_p$ clearly $a \cdot^{\m A} a = a \cdot_p a \le_p a$, so $a \cdot^{\m A} a \le^{\m A} a$.
\end{proof}

\begin{remark}
For the reader's convenience, let us record explicitly what it means for two families of maps $\varphi_{pq},\psi_{pq}:A_p\to A_q$ defined for $p\jle q$ in $\m I$ to induce a directed system of metamorphisms as in the above theorem: (i) $\varphi_{pp}=\psi_{pp}=\mathrm{id}_{A_p}$ for all $p \in I$, (ii) $\varphi_{qr}\circ\varphi_{pq}=\varphi_{pr}$ and $\psi_{qr}\circ\psi_{pq}=\psi_{pr}$ whenever $p\jle q\jle r$ in $\m I$, and (iii) for all $a,b\in A_p$
 \[
 \varphi_{pq}(ab)=\varphi_{pq}(a)\varphi_{pq}(b),\quad
 \psi_{pq}(a\ld b)=\varphi_{pq}(a)\ld \psi_{pq}(b)\quad\text{and}\quad
 \psi_{pq}(a\rd b)=\psi_{pq}(a)\rd \varphi_{pq}(b).
 \]
\end{remark}

  The structure theorems proved in this section extend the structure theorems for involutive po-monoids from~\cite{GFJiLo23}. An \emph{involutive po-monoid} is a structure of the form $\m A = \pair{A,\le,\cdot,1,\nein,\no}$ such that $\pair{A,\le}$ is a poset and $\pair{A,\cdot, 1}$ is a monoid satisfying
\[
x\le y\iff x\cdot\nein y\le \no 1 \iff \no y\cdot x\le \no 1.
\]
Such structures are in fact residuated monoids, with residuals defined by
\[
  x\ld y := \nein(\no y\cdot x), \qquad x\rd y := \no(y\cdot\nein x).
\]
An involutive po-monoid $\m A$ was called \emph{locally integral} in~\cite{GFJiLo23} if it is balanced, square-decreasing (i.e., satisfies $x \cdot x \le x$), and satisfies $x\ld (x \rd x) \approx x \rd x$. Such structures were decomposed in~\cite{GFJiLo23} as \emph{ordinary} P\l{}onka sums of directed systems of homomorphisms $\varphi_{pq}$ between the fibers $\m A_p$, with $\varphi_{pq}$ and $\m A_p$ defined in the exactly the same way as in the present paper. The additional conditions, namely involutivity and local integrality, made it possible to disregard the maps $\psi_{pq}$. It is, of course, possible to add these maps to the picture. The maps $\odot,\otimes\: A^2\to A$ defined by $a\odot b := 1_b\cdot a$ and $a\otimes b := a / 1_b = \no(1_b\cdot\nein a)$ can be shown to be homotactic left normal bands and the assignment $1\mapsto\odot$, $\cdot\mapsto\pair{\odot,\odot,\odot}$, $\nein\mapsto\pair{\otimes,\odot}$, and $\no\mapsto\pair{\otimes,\odot}$ is a partition system for each locally integral involutive po-monoid~$\m A$. Thus, $\m A$ is a P\l{}onka sum whose order can be recovered by~\eqref{eq:def:sum:order}. This is precisely the decomposition obtained in an ad hoc manner in~\cite{GFJiLo23}.

Let us now show how this class of ordered algebras fits into the present framework. Involutive po-monoids may equivalently be presented as residuated monoids equipped with a constant $0$ satisfying the equations $0 \rd (x \ld 0) \approx (0 \rd x) \ld 0$: in one direction we take $0 := \nein 1 = \no 1$, while in the other direction we take $\nein x := x \ld 0$ and $\no x := 0 \rd x$. Our main structure theorems may be modified to account for this extra constant as follows. In the decomposition result (Theorem~\ref{thm:steady:over:I:partition:meta}), if $\m A$ is an involutive po-monoid steady over $\m I \le \ZIDP \m A$ with $1 \in I$, then $0 \in A_1$ and each $\m A_p$ becomes an involutive po-monoid when equipped with the constant $0_{p} := \varphi_{1 p}(0) = p \cdot^{\m A} 0$. In the composition result (Theorem~\ref{thm:construction:meta}), if $\m I$ has a least element $\bot$ and each $\m A_p$ is an involutive po-monoid equipped with the constant $0_{p} = \varphi_{\bot p}(0_{\bot})$, then the sum $\m A$ is an involutive pomonoid when equipped with the constant $0_{\bot}$. Moreover, we already saw that in both the composition and the decomposition result, the integrality of the fibers corresponds to the square-decreasing property of the sum.

\section{Balanced Idempotent Residuated Semigroups}
\label{sec:idempotent case}

In this section we specialize the structural description of balanced residuated semigroups to the
idempotent case, where they turn out to be commutative (Proposition~\ref{prop:idemcomm}). In the lattice-ordered case, the fibers are Brouwerian algebras, i.e.,  commutative idempotent integral residuated lattices, and the maps $\varphi_{pq}$ preserve binary meets (Proposition~\ref{prop:phi:preserves:meet}). Since multiplication and meet coincide in the idempotent integral case, these fibers are in fact distributive lattices.

\begin{proposition}\label{prop:idemcomm}
    The following conditions are equivalent for each idempotent residuated semigroup~$\m A$:
    \begin{enumerate}[(1)]
        \item $\m A$ satisfies $1_x \cdot y \approx y\cdot 1_x$,
        \item $\m A$ satisfies $x\ld x \approx x\rd x$,
        \item $\m A$ is commutative.
    \end{enumerate}
    Moreover, each of these conditions implies condition~\ref{cond:H}.
\end{proposition}

\begin{proof}
First of all, let us prove that~(1) implies condition~\ref{cond:H}. By Proposition~\ref{prop:H2:iff:H3:imply:H1}, it suffices to prove that~\ref{H2} holds. Let $a$ and $b$ be arbitrary elements and assume $1_a = 1_b$. First of all, by the centrality of $1_a$, we have $1_a\le a\ld a$; and by Lemma~\ref{lem:inequalities:selfresiduals:central}, we have that $1_a \le 1_{a\rd b}$. In order to prove the reverse inequality, notice that idempotence implies that $b \le 1_b = 1_a\le a\ld a$ and therefore $ab \le a$, whence we deduce that $a\le a\rd b$ and therefore $a = aa \le a(a\rd b)$. But also $a \le 1_a = 1_b$, so by idempotence $a\rd b \le 1_b\rd b = 1_b = 1_a \le a\ld a$, whence $a(a\rd b) \le a$. That is, $a(a\rd b) = a$. Thus,
\[
1_{a\rd b}\cdot a = 1_{a\rd b}\cdot a(a\rd b) =  a 1_{a\rd b}\cdot (a\rd b) \le a (a\rd b) = a,
\]
whence $1_{a\rd b} \le 1_a$.

\begin{itemize}[topsep=0pt, align=left, left=0pt .. \parindent]
\item[(1) $\Rightarrow$ (2)] We have seen that~(1) implies condition~\ref{cond:H}, which in turn implies $x\ld x \approx x\rd x$ by Proposition~\ref{prop:H2:iff:H3:imply:H1}.

\item[(2) $\Rightarrow$ (3)] Idempotence implies that for all elements $a$ and $b$, we have $bba = ba$, hence $b\le ba\rd ba = ba\ld ba$ and $bab\le ba$. Multiplying by $a$ gives $abab\le aba$, so by idempotence we obtain $ab\le aba$. Finally, $baa = ba$ implies $a\le ba\ld ba = ba\rd ba$, hence $aba\le ba$ and therefore $ab \le ba$.

\item[(3) $\Rightarrow$ (1)] This is trivial.\QED
\end{itemize}\let\QED\relax
\end{proof}

The next result shows that right-residuation, idempotence, and $x\rd x \approx y\rd y$ imply $y \approx (x\rd x)y$ (even in the nonassociative setting). Hence if all self-residuals coincide (i.e., $x\ld x\approx y\rd y$) in an idempotent residuated semigroup, then it is balanced and the term $1_x=x\rd x$ is the identity and top element.

\begin{lemma}\label{lem:identity:and:top}
    Suppose $\m A = \pair{A,\le,\cdot,\rd }$ satisfies the identity $xx \approx x$ and $\rd $ is a right residual of the multiplication, i.e., $xy\le z\iff x\le z\rd y$. Then $\m A$ satisfies $x \approx (x\rd x)x$. If $\m A$ moreover satisfies $x\rd x \approx y\rd y$, then $y \approx (x\rd x)y$ and $y\le x\rd x$ also hold.
\end{lemma}

\begin{proof}
Let $a,b,c\in A$ and assume $a\le b$. Using right residuation, we have $bc\le bc$ implies $a\le b\le (bc)\rd c$, and therefore $ac\le bc$. Hence, multiplication is monotone in the left coordinate. Therefore, since $bb = b$ implies that $b \le b\rd b$, we deduce that $b = bb \le (b\rd b) b \le b$. The second part follows immediately.
\end{proof}

Recall that a Brouwerian semilattice $\pair{A,\wedge,\to,1}$ is a meet-semilattice with a top element $1$ such that $x\wedge y\le z\iff y\le x\to z$.
A Brouwerian algebra is a Brouwerian semilattice expanded with a join $\vee$, and it follows from the residuation equivalence that the lattice reduct is distributive.

\begin{corollary}\label{cor:Brsl}
    Every idempotent residuated semigroup that satisfies the equation $x\ld x \approx y\rd y$ is a Brouwerian semilattice with $x\to y=x\ld y$ and $1=x\ld x$. In particular, every integrally closed idempotent residuated monoid is a Brouwerian semilattice.
\end{corollary}

\begin{proof}
Assume the identities hold in an idempotent residuated semigroup. Let $a,b$ be arbitrary elements. From $a\le a \ld a = b \rd  b$ we conclude $ab\le b$. By Proposition~\ref{prop:idemcomm}, the multiplication operation is commutative, hence can be considered a meet-semilattice operation. If $a\le b$, then $ab = a$, since $a = aa\le ab = ba \le a$. On the other hand, if $ab = a$, then $a = ab \le b$, so $\pair{A,\le}$ is a meet-semilattice and the multiplication operation $x y$ coincides with the meet $x \wedge y$ in $\le$. Moreover, the residual $x\to y:=x\ld y$ is then a Brouwerian semilattice implication. The fact that $1 = a\ld a = a\rd a$ is the identity and top element follows from Lemma~\ref{lem:identity:and:top}.
\end{proof}

Hence every balanced idempotent residuated semigroup has fibers that are closed under multiplication and residuation, and Corollary~\ref{cor:Brsl} above shows that these fibers are Brouwerian semilattices. The structure theory of the previous section applies to all balanced idempotent residuated semigroups that are steady, i.e., satisfy \ref{St2}: $1_{x\to y}=1_x1_y$. 
The last result in this section shows that for commutative idempotent residuated lattices the map $\varphi_{pq}(x):=xq$ from $\m A_p$ to $\m A_q$ is meet preserving. 

\begin{proposition} \label{prop:phi:preserves:meet}
    The following quasi-identity holds in each idempotent commutative residuated lattice:
    \[
    1_x \approx 1_y\le 1_z \implies (x\wedge y)1_z \approx x1_z\wedge y1_z.
    \]
\end{proposition}

\begin{proof}
Let $\m A$ be an idempotent commutative residuated lattice. The fibers of $\m A$ are Brouwerian algebras, hence $ab =a\wedge b$ if $1_a = 1_b$. Similarly, for any $q = 1_c$ greater than $1_a = 1_b$, $aq$ and $bq$ are in the same component and satisfy $aqbq=aq\wedge bq$. The conclusion follows by centrality and idempotence of $q$.
\end{proof}

\section{Instructive Examples}
\label{sec:examples}

We recall several instructive instances of P\l{}onka sum decompositions already considered in~\cite{BGJPS2024}.

\paragraph{\bfseries P\l{}onka sum of two residuated posets.}
Consider two residuated posets $\m A_1$ and $\m A_2$ and two directed systems $\Phi=\{\varphi_{pq} : p\le q\}$ and $\Psi=\{\psi_{pq} : p\le q\}$, indexed over the 2-element chain $1<2$, such that the nonidentity maps $\varphi_{12},\psi_{12}\: A_1\to A_2$ are defined by
\[
a\mapsto\varphi_{12}(a) = 1^{\m A_2}\quad\text{and}\quad a\mapsto\psi_{12}(a) = 0^{\m A_2},
\]
with $0^{\m A_2}$ a fixed element in $A_2$ such that $0^{\m A_2} < 1^{\m A_2}$. Then $\pair{\Phi,\Psi}$ satisfies the conditions of  Theorem~\ref{thm:construction:meta}, and hence we obtain a residuated poset $\m S = \m A_1\uplus\m A_2$ (Figure~\ref{fig:residuatedposetsum}).

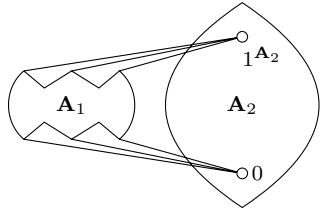
\begin{figure}[ht]
\centering
\begin{tikzpicture}[every label/.append style={font=\scriptsize}]
\node at (0,3)[n]{\scriptsize$\m A_1$};
\node at (5,3)[n]{\scriptsize$\m A_2$};

\draw (5,0)..controls(2,2)and(2,4)..(5,6)..controls(8,4)and(8,2)..(5,0)
(-1.4,4)--(-.8,3.5)--(0,4)--(.8,3.5)--
(1.4,4)..controls(2,3.5)and(2,2.5)..(1.4,2)
(-1.4,4)..controls(-2,3.5)and(-2,2.5)..(-1.4,2)--
(-.8,2.5)--(0,2)--(.8,2.5)--(1.4,2)--(5,1)(0,2)--(5,1)(-1.4,2)--(5,1)
(1.4,4)--(5,5)(0,4)--(5,5)(-1.4,4)--(5,5)
(5,1) node(0p)[label=+45:$0^{\m A_2}$]{}
(5,5) node(1p)[label=-45:$1^{\m A_2}$]{};
\end{tikzpicture}
\caption{Example of a P\l{}onka sum $\m A_1\uplus\m A_2$ of two residuated posets.}
\label{fig:residuatedposetsum}
\end{figure}

\paragraph{\bfseries Doubly chopped lattices.}
A \emph{doubly chopped lattice} (see~\cite{GratzerSchmidt95}, \cite{GratzerBook}) is a poset in which every pair of elements with an upper bound has a join and every pair of elements with a lower bound has a meet. Such a poset will become a lattice if a new top and bottom element is added to the poset, and conversely, every doubly chopped lattice is obtained from a bounded lattice by removing the bounds.

If $\m A$ is a residuated poset that is a doubly chopped lattice and ${\m B}$ is a residuated lattice, then the P\l{}onka sum $\m S = \m A\uplus\m B$ is a lattice under the following operations:
{\small
\begin{align*}
a\vee^{\m S} b &= \begin{cases}
  a\vee^{{\m A}} b& \text{ if } a,b\in A\text{ have an upper bound in $A$}\\
  1^\m B& \text{ if } a,b\in A\text{ have no upper bound in $A$}\\
  a\vee^{{\m B}} b& \text{ if } a,b\in B \\
  a &\text{ if } a\in A, \;b\in B \text{ with } b\le 0^{\m B} \\
  b\vee^{{\m B}} 1^{{\m B}} &\text{ if } a\in A, \;b\in B \text{ with } b\not\le 0^{\m B},
\end{cases} \\
a\wedge^{\m S} b &= \begin{cases}
  a\wedge^{{\m A}} b&\text{ if } a,b\in A \text{ have a lower bound in $A$}\\
  0^{\m B}&\text{ if } a,b\in A \text{ have no lower bound in $A$}\\
  a\wedge^{{\m B}} b& \text{ if } a,b\in B \\
  a &\text{ if } a\in A, \;b\in B \text{ with } 1^{{\m B}}\le b \\
  b\wedge^{{\m B}} 0^{\m B} &\text{ if } a\in A, \;b\in B \text{ with } 1^{{\m B}}\not\le b.
\end{cases}
\end{align*}
}

It is easily checked that $a,b\le a\vee^{\m S} b$. Moreover, consider $a,b\le^{\m S} c$, for some $c\in A\cup B $, $a\in A$, and $b\in B$ (the other cases are trivial). If $c\in A$, then $b\le^{\m S} c$ implies $b\le^{\m B} 0^{\m B}$, then $a\vee^{\m S}b = a\le c$. Alternatively, if $c\in B$, then $a\le^\m S c$ implies $1^{\m B}\le c$, then $a\vee^{\m S}b =1^{\m B}\vee^{\m S}b\le c $. The case of $\wedge^{\m S} $ can be checked analogously.

\paragraph{\bfseries The relation algebra $\mathcal P(\mathbb Z_2)$.}
For any monoid $\m M =\pair{M,\cdot,e}$ the \emph{complex algebra} $\mathcal P(\m M)$ is the residuated lattice $\pair{\mathcal P(M),\cap,\cup,\cdot,\ld,\rd,\{e\}}$, where for all $X,Y\subseteq M$, $X\cdot Y=\{xy\: x\in X, y\in Y\}$, $X\ld Y=\{z\in M\: X\cdot\{z\}\subseteq Y\}$ and $X\rd Y=\{z\in M\: \{z\}\cdot Y\subseteq X\}$. The P\l{}onka sum with two components presented above (Figure~\ref{fig:residuatedposetsum}) encompasses the example of the 4-element relation algebra $\mathcal P(\mathbb Z_2)$ given by the complex algebra of the 2-element group (Figure~\ref{fig:a2}). Note that this relation algebra is not locally integral~\cite{GFJiLo23}, but it is balanced and satisfies the identities~\ref{St1}--\ref{St3}, hence the P\l{}onka sum decomposition can be applied to all members of the variety of relation algebras generated by $\mathcal P(\mathbb Z_2)$.

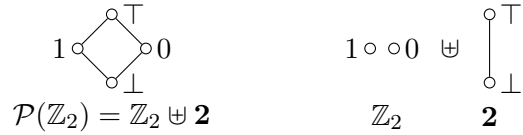
\begin{figure}
\centering
\begin{tikzpicture}[baseline=0pt]
\node at (0,-1)[n]{$\mathcal P(\mathbb Z_2)=\mathbb Z_2\uplus\mathbf 2$};
\node(3) at (0,2)[label=right:$\top$]{};
\node(2) at (1,1)[label=right:$0$]{} edge (3);
\node(1) at (-1,1)[label=left:$1$]{} edge (3);
\node(0) at (0,0)[label=right:$\bot$]{} edge (1) edge (2);
\end{tikzpicture}
\qquad \qquad
\begin{tikzpicture}[baseline=0pt]
\node at (-2,-1)[n]{$\mathbb Z_2$};
\node at (1,-1)[n]{$\m2$};
\node(2) at (-1.8,1)[label=right:$0$]{};
\node(1) at (-2.5,1)[label=left:$1$]{};
\node(1) at (-.15,1)[n]{$\uplus$};
\node(3) at (1,2)[label=right:$\top$]{};
\node(0) at (1,0)[label=right:$\bot$]{} edge (3);
\end{tikzpicture}

\caption{A 4-element relation algebra obtained from the P\l{}onka sum of the 2-element group and the 2-element Boolean algebra.}
\label{fig:a2}
\end{figure}    

\section{Conclusion}
\label{sec:conclusion}

The main universal algebraic contribution of this paper was to generalize the P\l{}onka sum construction to multiple partition functions, or equivalently to multiple families of linking maps, which allows for composing and decomposing larger classes of algebras. Unlike the ordinary P\l{}onka sum construction, this generalization is moreover well-suited for constructing ordered algebras. We then applied this construction to show that a wide class of residuated semigroups, which we called steady residuated semigroups, can be decomposed in this way into a family of integrally closed residuated submonoids, indexed by the positive idempotents of the original residuated semigroup. This extends more specialized existing constructions, which decompose residuated semigroups into integral involutive pieces. The classes of ordered algebras studied in this paper are summarized in Table~\ref{tab:residuated_structures_summary}. Apart from the ordered quasivariety of balanced residuated semigroups with \ref{cond:H}, all of these classes are partially ordered varieties. That is, they are axiomatized by a set of inequalities together with a specification of the monotonicity or antitonicity properties of the primitive operations.

\begin{table}
\centering
\small
\tabcolsep5pt
\begin{tabular}{p{1.9cm} p{6.1cm} p{6.4cm}}
\hline
\textbf{Class} & \textbf{Defining Axioms} & \textbf{Decomposition and Fibers} \\ \hline
\textbf{Residuated\newline semigroups} & Poset $\langle A, \le \rangle$ and semigroup $\langle A, \cdot \rangle$ \newline satisfying residuation: \newline $x \cdot y \le z \ \Leftrightarrow \ x \le z/y \ \Leftrightarrow \ y \le x\backslash z$ & None in general. \\ \hline
\textbf{Balanced\newline residuated\newline semigroups} & Residuated semigroup satisfying: \newline $x\backslash x \approx x/x$ (denoted $1_x$) and \newline $y \le 1_x\cdot y$ and $y \le y\cdot 1_x$ \newline(self-residuals are positive) & Partitioned into disjoint sets $A_p$ indexed by positive idempotents $\Idp\mathbf{A}$, but $A_p$ are not necessarily subalgebras. \\ \hline
\textbf{Balanced\newline residuated\newline semigroups\newline with \ref{cond:H}} & Balanced residuated semigroup satisfying: \newline \ref{H1} $1_x \approx 1_y \implies 1_{xy} \approx 1_x$ \newline \ref{H2} $1_x \approx 1_y \implies 1_{x/y} \approx 1_x$ \newline \ref{H3} $1_x \approx 1_y \implies 1_{x\backslash y} \approx 1_x$ & Partitioned into disjoint subalgebras $\mathbf{A}_p$ {($p\in\Idp \mathbf{A}$)} which are integrally closed ($x \backslash x \approx 1 \approx x / x$) residuated monoids. \\ \hline
\textbf{Steady\newline residuated\newline semigroups} & Balanced residuated semigroup \mbox{satisfying}: \newline \ref{St1} $1_{xy} \approx 1_x \cdot 1_y$ \newline \ref{St2} $1_{x/y} \approx 1_x \cdot 1_y$ \newline \ref{St3} $1_{y \backslash x} \approx 1_x \cdot 1_y$ & Generalized P\l{}onka sum of a directed system of metamorphisms of integrally closed residuated monoids over the semilattice $\IDP \mathbf{A}$. \\ \hline
\textbf{Idempotent\newline steady\newline residuated\newline semigroups} & Steady residuated semigroup satisfying: \newline $x x \approx x$ (idempotence) and \newline $1_x \cdot y \approx y \cdot 1_x$\newline  (equivalently, commutativity) & Decomposes into Brouwerian semilattices. \\ \hline
\end{tabular}
\caption{Summary of classes of residuated structures, their defining axioms, and decomposition results.}
\label{tab:residuated_structures_summary}
\end{table}

The basic structure theorems for steady residuated semigroups proved in this paper form a promising basis for further investigations. For example, when is it the case that the construction yields lattice-ordered structures? What are some natural classes of residuated semigroups for which it gives new insights into their structure theory? And how far beyond the class of steady residuated semigroups can we go if we iterate the P\l{}onka sum construction, as in Example~\ref{ex:iterated:decomposition}?

\section*{Acknowledgments}

S. Bonzio acknowledges the support by the Italian Ministry of Education, University and Research through the projects PRIN 2022 DeKLA (``Developing Kleene Logics and their Applications'', code: 2022SM4XC8) and PRIN Pnrr project Qm4Np (``Quantum Models for Logic, Computation and Natural Processes'', code: P2022A52CR). He also acknowledges the Fondazione di Sardegna for the support received by the project MAPS (grant number F73C23001550007). Finally, he gratefully acknowledges also the support of the INDAM GNSAGA (Gruppo Nazionale per le Strutture Algebriche, Geometriche e loro Applicazioni). The work of A. P\v{r}enosil was funded by the grant 2021 BP 00212 of the grant agency AGAUR of the Generalitat de Catalunya. S. Bonzio and A. P\v{r}enosil also acknowledge the support of the MOSAIC project (H2020-MSCA-RISE-2020 Project 101007627), which funded their visits to Chapman University.

\bibliographystyle{fundam}

\end{document}